\documentclass[journal]{IEEEtran}
\usepackage[utf8]{inputenc}

\usepackage{amsmath,amssymb,amsfonts}

\usepackage{amsthm}
\usepackage{graphicx}
\usepackage{textcomp}
\usepackage{xcolor}
\usepackage[bookmarks=true]{hyperref}
\usepackage{subfigure}

\usepackage[algo2e,ruled,linesnumbered,vlined]{algorithm2e}

\newtheorem{problem}{Problem}
\newtheorem{theorem}{Theorem}
\newtheorem{remark}{Remark}

\def\newsym #1[#2]#3{\expandafter\newcommand\csname sym:#1\endcsname[#2]{#3}}
\def\s #1{\csname sym:#1\endcsname}

\usepackage{amsmath}

\DeclareMathOperator{\union}{\bigcup} 
\newcommand{\trans}{\ensuremath{\mathsf{T}}} 
\newcommand{\Call}[2]{{\sc#1}(#2)} 
\SetKwProg{Fn}{function}{}{end} 
\SetKwFor{While}{while}{}{end}
\SetKwIF{If}{ElseIf}{Else}{if}{}{elseif}{else}{end}%
\SetKwIF{uIf}{uElseIf}{uElse}{if}{}{elseif}{else}{}%
\SetKwComment{Comment}{\%}{}

\newsym{lrp}[1]{\left({#1}\right)}
\newsym{lrcb}[1]{\left\{{#1}\right\}}
\newsym{lrsb}[1]{\left[{#1}\right]}
\newsym{set/rel/diff}[0]{\setminus}
\newsym{set/rel/subeq}[0]{\subseteq}
\newsym{set/rel/sub}[0]{\subset}
\newsym{set/rel/supeq}[0]{\supseteq}
\newsym{set/rel/sup}[0]{\supset}
\newsym{set/rel/not/subeq}[0]{\not\subseteq}
\newsym{set/rel/not/sub}[0]{\not\subset}
\newsym{set/rel/not/supeq}[0]{\not\supseteq}
\newsym{set/rel/not/sup}[0]{\not\supset}
\newsym{set/rel/bnd}[1]{\partial#1}
\newsym{set/rel/vol}[1]{\mathrm{vol}\s{lrp}{#1}}
\newsym{set/rel/vol/scaled}[2]{\mathrm{vol}_{#2}\s{lrp}{#1}}
\newsym{set/rel/union}[0]{\cup}
\newsym{set/rel/union/multi}[3]{\bigcup\limits_{#1}^{#2}{#3}}
\newsym{set/rel/isect}[0]{\cap}
\newsym{set/rel/isect/multi}[3]{\bigcap\limits_{#1}^{#2}{#3}}
\newsym{set/rel/cmpl}[1]{\mathrm{cmlp}{#1}}
\newsym{set/num/real}[0]{\mathbb{R}}
\newsym{set/num/nat}[0]{\mathbb{N}}
\newsym{set/num/circle}[0]{\mathbb{S}^{1}}
\newsym{set/fun/cont}[1]{\mathcal{C}^{#1}}
\newsym{set/empty}[1]{\emptyset}
\newsym{set/Realsets}[0]{\mathfrak{S}}

\newsym{set/enum}[1]{\s{lrcb}{#1}}
\newsym{set/def}[2]{\s{lrcb}{#1~|~#2}}
\newsym{set/idx}[1]{\mathfrak{I}_{#1}}

\newsym{mtx/eye}[0]{I}
\newsym{mtx/row}[2]{{#1}_{[#2]}}

\newsym{prod}[2]{\prod_{#1}^{#2}}
\newsym{prod/t}[2]{\sideset{}{^T}\prod_{#1}^{#2}}

\newsym{abs}[1]{ {\left| {#1} \right|} }
\newsym{norm/1}[1]{\left\|{#1}\right\|_{1}}
\newsym{norm/2}[1]{\left\|{#1}\right\|_{2}}
\newsym{norm/p}[2]{\left\|{#1}\right\|_{#2}}
\newsym{norm/max}[1]{\left\|{#1}\right\|_{\mathrm{max}}}
\newsym{norm/rw/p}[2]{\left\|{#1}\right\|_{\mathrm{R,#2}}}

\newsym{ws}[0]{\mathcal{W}}

\newsym{ws/obst}[1]{\mathcal{O}_{#1}}
\newsym{ws/obst/bnd}[1]{\s{set/rel/bnd}{\mathcal{O}_{#1}}}
\newsym{ws/obst/idxset}[0]{\s{set/idx}{\mathcal{O}}}
\newsym{ws/obst/amount}[0]{N_{\mathcal{O}}}

\newsym{ws/bnd}[1]{\s{set/rel/bnd}{\s{ws}}_{#1}}
\newsym{ws/bnd/out}[0]{\s{ws/bnd}{0}}
\newsym{ws/bnd/ins}[1]{\s{ws/bnd}{#1}, #1 \in \s{ws/obst/idxset}}

\newsym{ws/compl}[0]{\mathrm{cmpl}(\s{ws})}

\newsym{frame/ws}[0]{\mathcal{F}_{\s{ws}}}
\newsym{frame/rob}[0]{\mathcal{F}_{\s{rob/body}}}

\newsym{rob/body}[0]{\mathcal{R}}
\newsym{rob/body/at}[1]{\s{rob/body}\s{lrp}{#1}}
\newsym{rob/body/bnd}[0]{\s{set/rel/bnd}{\s{rob/body}}}
\newsym{rob/cfg}[0]{p}
\newsym{rob/pos/x}[0]{x}
\newsym{rob/pos/y}[0]{y}
\newsym{rob/angle}[0]{\theta}
\newsym{rob/orient}[0]{\hat{n}}
\newsym{rob/orient/at}[1]{\s{rob/orient}\s{lrp}{#1}}

\newsym{rob/pos/x/lb}[0]{\s{rob/pos/x}_{l}}
\newsym{rob/pos/x/ub}[0]{\s{rob/pos/x}_{u}}
\newsym{rob/pos/y/lb}[0]{\s{rob/pos/y}_{l}}
\newsym{rob/pos/y/ub}[0]{\s{rob/pos/y}_{u}}
\newsym{rob/angle/lb}[0]{\s{rob/angle}_{l}}
\newsym{rob/angle/ub}[0]{\s{rob/angle}_{u}}

\newsym{rob/state}[0]{z}
\newsym{rob/state/dot}[0]{\dot{\s{rob/state}}}
\newsym{rob/state/at}[1]{\s{rob/state}_{#1}}
\newsym{rob/state/cfg}[0]{\s{rob/cfg}}
\newsym{rob/state/misc}[0]{q}
\newsym{rob/state/misc/ith}[1]{\s{rob/state/misc}_{\s{lrsb}{#1}}}
\newsym{rob/state/misc/lb}[0]{\s{rob/state/misc}_{\ell}}
\newsym{rob/state/misc/lb/ith}[1]{{\s{rob/state/misc/lb}}_{\s{lrsb}{#1}}}
\newsym{rob/state/misc/ub}[0]{\s{rob/state/misc}_{u}}
\newsym{rob/state/misc/ub/ith}[1]{{\s{rob/state/misc/ub}}_{\s{lrsb}{#1}}}
\newsym{rob/input}[0]{u}
\newsym{rob/input/at}[1]{\s{rob/input}_{#1}}
\newsym{rob/state/dim}[0]{n}
\newsym{rob/input/dim}[0]{m}

\newsym{rob/ss}[0]{\mathcal{Z}}
\newsym{rob/ss/cfg}[0]{\s{rob/ss}_{[\s{rob/state/cfg}]}}
\newsym{rob/ss/misc}[0]{\s{rob/ss}_{[\s{rob/state/misc}]}}
\newsym{rob/ss/safe}[0]{\s{rob/ss}_{s}}
\newsym{rob/ss/safe/oa}[0]{\overline{\s{rob/ss}}_{s}}
\newsym{rob/ss/safe/ua}[0]{\underline{\s{rob/ss}}_{s}}
\newsym{rob/ss/cfg/safe}[0]{\s{rob/ss}_{[\s{rob/state/cfg}],s}}
\newsym{rob/ss/cfg/safe/oa}[0]{%
        \overline{\s{rob/ss}}_{[\s{rob/state/cfg}],s}}
\newsym{rob/is}[0]{\s{set/num/real}^{\s{rob/input/dim}}}

\newsym{rob/ss/misc/compl}[0]{\mathrm{cmpl}\s{lrp}{\s{rob/ss/misc}}}

\newsym{rob/f}[0]{f}
\newsym{rob/f/at}[2]{\s{rob/f}\s{lrp}{{#1}, {#2}}}
\newsym{rob/dynamics}[2]{\s{rob/f/at}{#1}{#2}}

\newsym{rob/f/oa}[0]{F}
\newsym{rob/f/oa/at}[2]{\s{rob/f/oa}\s{lrp}{{#1}, {#2}}}

\newsym{rob/state/goal}[0]{\s{rob/state}^{\star}}
\newsym{rob/cfg/goal}[0]{\s{rob/cfg}^{\star}}

\newsym{rob/ss/partition}[0]{\mathcal{P}}
\newsym{rob/ss/partition/safe}[0]{\s{rob/ss/partition}_{S}}
\newsym{rob/ss/partition/frs}[0]{\mathfrak{R}(\s{rob/ss/partition})}
\newsym{rob/ss/frontier}[0]{\mathcal{F}}
\newsym{rob/ss/partition/wthresh}[0]{\epsilon_{w}}
\newsym{rob/ss/refine/vphresh}[0]{\epsilon_{V_{\s{rob/state/cfg}}}}
\newsym{rob/ss/refine/vqhresh}[0]{\epsilon_{V_{\s{rob/state/misc}}}}

\newsym{rob/ss/cell}[0]{\mathcal{C}}

\newsym{rob/ss/cfg/cell}[0]{\mathcal{C}_{[\s{rob/state/cfg}]}}
\newsym{rob/ss/cfg/cell/at}[1]{\s{rob/ss/cfg/cell}_{#1}}
\newsym{rob/ss/cfg/cell/of}[1]{\s{rob/ss/cfg/cell}_{#1}}
\newsym{rob/ss/cfg/cell/of/ith}[2]{\s{rob/ss/cfg/cell}_{#1,[#2]}}

\newsym{rob/ss/cfg/cell/oa}[0]{\overline{\s{rob/ss/cfg/cell}}}
\newsym{rob/ss/cfg/cell/ua}[0]{\underline{\s{rob/ss/cfg/cell}}}

\newsym{rob/ss/misc/cell}[0]{\mathcal{C}_{[\s{rob/state/misc}]}}

\newsym{rob/fpoa}[0]{\overline{\s{rob/body}}}
\newsym{rob/fpua}[0]{\underline{\s{rob/body}}}
\newsym{rob/fpoa/at}[1]{\s{rob/fpoa}\s{lrp}{#1}}
\newsym{rob/fpua/at}[1]{\s{rob/fpua}\s{lrp}{#1}}

\newsym{rob/ss/frs/crit}[0]{\s{rob/ss}_{c}}
\newsym{rob/ss/frs/crit/oa}[0]{\overline{\s{rob/ss}}_{c}}

\newsym{rob/ss/cell/frs}[0]{\mathfrak{R}}
\newsym{rob/ss/cell/frs/at}[1]{\s{rob/ss/cell/frs}\s{lrp}{#1}}

\newsym{ml/base/costfun}[0]{J}
\newsym{ml/base/costfun/at}[1]{\s{ml/base/costfun}\s{lrp}{#1}}
\newsym{ml/base/coef/err}[0]{\lambda_{B}}
\newsym{ml/base/coef/reg}[0]{\lambda_{R}}
\newsym{ml/base/reg}[0]{r}
\newsym{ml/base/reg/at}[1]{\s{ml/base/reg}\s{lrp}{#1}}

\newsym{ml/aug/costfun}[0]{J^{\star}}
\newsym{ml/aug/costfun/at}[1]{\s{ml/aug/costfun}\s{lrp}{#1}}
\newsym{ml/aug/coef/safety}[0]{\lambda_{S}}
\newsym{ml/aug/safety/penalty}[0]{h}
\newsym{ml/aug/safety/penalty/at}[1]{\s{ml/aug/safety/penalty}\s{lrp}{#1}}
\newsym{ml/aug/metric}[0]{V}
\newsym{ml/aug/metric/at}[1]{\s{ml/aug/metric}\s{lrp}{#1}}

\newsym{ml/aug/epochs/amount}[0]{N_{\mathrm{epochs}}}

\newsym{data/set}[0]{\mathcal{D}}
\newsym{data/point}[0]{d}
\newsym{data/point/state}[0]{\s{rob/state}}
\newsym{data/point/input}[0]{\s{rob/input}}
\newsym{data/point/ith}[1]{\s{data/point}_{#1}}
\newsym{data/point/ith/state}[1]{\s{data/point/state}_{#1}}
\newsym{data/point/ith/input}[1]{\s{data/point/input}_{#1}}

\newsym{nn}[0]{\phi}
\newsym{nn/at}[1]{\phi\s{lrp}{#1}}
\newsym{nn/in/dim}[0]{\s{rob/state/dim}}
\newsym{nn/out/dim}[0]{\s{rob/input/dim}}
\newsym{nn/weight}[0]{w}
\newsym{nn/bias}[0]{b}
\newsym{nn/layers/amount}[0]{N_{\phi}}
\newsym{nn/layers/idxset}[0]{\s{set/idx}{\s{nn/layers/amount}}}
\newsym{nn/layer/size}[1]{n_{#1}}
\newsym{nn/layer/weight}[1]{w_{#1}}
\newsym{nn/layer/weight/row}[2]{w_{#1[#2]}}
\newsym{nn/layer/bias}[1]{b_{#1}}
\newsym{nn/layer/bias/row}[2]{b_{#1[#2]}}
\newsym{nn/layer/afun}[1]{h_{#1}}
\newsym{nn/layer/afun/at}[2]{\s{nn/layer/afun}{#1}{\left({#2}\right)}}
\newsym{nn/layer/afun/row}[2]{h_{#1[#2]}}
\newsym{nn/layer/afun/row/at}[3]{\s{nn/layer/afun/row}{#1}{#2}\s{lrp}{#3}}
\newsym{nn/layer/map}[1]{\phi_{#1}}
\newsym{nn/layer/map/at}[2]{\s{nn/layer/map}{#1}{\left(#2\right)}}

\newsym{nn/oa}[0]{\Phi}
\newsym{nn/oa/at}[1]{\s{nn/oa}\s{lrp}{#1}}

\newsym{nn/layer/map/rcrs}[1]{z_{#1}}
\newsym{nn/layer/map/rcrs/at}[2]{\s{nn/layer/map/rcrs}{#1}\s{lrp}{#2}}
\newsym{nn/layer/afdm}[1]{H_{#1}}
\newsym{nn/layer/afdm/at}[2]{\s{nn/layer/afdm}{#1}\s{lrp}{#2}}
\newsym{nn/layer/afdm/drv}[2]{H^{(#2)}_{#1}}
\newsym{nn/layer/afdm/drv/at}[3]{\s{nn/layer/afdm/drv}{#1}{#2}\s{lrp}{#3}}
\newsym{nn/layer/afdm/drv/1}[1]{H^{(1)}_{#1}}
\newsym{nn/layer/afdm/drv/1/at}[2]{\s{nn/layer/afdm/drv/1}{#1}\s{lrp}{#2}}
\newsym{nn/layer/afdm/drv/2}[2]{H^{(2)}_{#1,#2}}
\newsym{nn/layer/afdm/drv/2/at}[3]{\s{nn/layer/afdm/drv/2}{#1}{#2}\s{lrp}{#3}}

\begin{document}
\title{Failing with Grace: Learning Neural Network Controllers
  that are Boundedly Unsafe}
\author{Panagiotis~Vlantis, Leila J. Bridgeman, and Michael M. Zavlanos
\thanks{*This work is supported in part by AFOSR under award \#FA9550-19-1-0169, by ONR under agreement \#N00014-18-1-2374, and by NSF under award CNS-1932011.

Panagiotis~Vlantis, Leila J. Bridgeman, and Michael M. Zavlanos are with the Department of Mechanical
Engineering and Materials Science, Duke University, Durham, NC, USA.
Email: \{panagiotis.vlantis, leila.bridgeman, michael.zavlanos\}@duke.edu}%
}

\maketitle

\begin{abstract}%
 In this work,
we consider the problem of learning a feed-forward neural network controller
to safely steer an arbitrarily shaped planar robot in
a compact and obstacle-occluded workspace.
Unlike existing methods that depend strongly on the density of data points
close to the boundary of the safe state space to
train neural network controllers with closed-loop safety guarantees,
here we propose an alternative approach that lifts such strong assumptions on
the data that are hard to satisfy in practice and
instead allows for {\em graceful} safety violations, i.e., of a bounded magnitude that can be spatially controlled.
To do so,
we employ reachability analysis techniques to
encapsulate safety constraints in the training process.
Specifically,
to obtain a computationally efficient over-approximation of
the forward reachable set of the closed-loop system, 
we partition the robot's state space into cells and adaptively subdivide
the cells that contain states which may escape the safe set
under the trained control law.
%
%
Then,
using the overlap between each cell's forward reachable set and
the set of infeasible robot configurations as a measure for safety violations,
we introduce appropriate terms into the loss function that
penalize this overlap in the training process.
As a result,
our method can learn a safe vector field for the closed-loop system
and,
at the same time, provide worst-case bounds on
safety violation over the whole configuration space,
defined by the overlap between
the over-approximation of the forward reachable set of the closed-loop system
and
the set of unsafe states.
Moreover,
it can control the tradeoff between
computational complexity and tightness of these bounds.
Our proposed method is supported by both theoretical results and simulation studies.


\end{abstract}


\section{Introduction}%
\label{sec/intro}

In recent years,
progress in the field of machine learning has furnished
a new family of neural network controllers for robotic systems
that have significantly simplified the overall design process.
Robot navigation is one application
where neural network controllers have been successfully employed for
steering a variety of robots in static and dynamic environments~%
\cite{21113460,47372225,65705677,13285740,80553527}.
%
%
As such control schemes become more common in real-world applications,
the ability to train neural networks with safety considerations becomes
a necessity.

The design of reliable data-driven controllers
that result in safe and adaptable closed-loop systems has
typically relied on methods that couple
state-of-the-art machine learning algorithms
with control~%
\cite{83024768,11417186}.
A popular approach that falls in this class of methods employs
control barrier functions to appropriately constrain
the control inputs computed using data-driven techniques
so that a specified safe subset of the state space
remains invariant during execution and learning.
For example,
\cite{48876638} employs control barrier functions to
safeguard the exploration and deployment of
a reinforcement learning policy trained to achieve
high-level temporal specifications.
Similarly,
\cite{74521314} pairs
adaptive model learning with barrier certificates for
systems with possibly non-stationary agent dynamics.
Safety during learning is accomplished in~\cite{01373512}
by combining model-free reinforcement learning with barrier certificates and using 
Gaussian Processes to model the systems dynamics.
Although barrier certificates constitute
an intuitive tool for ensuring safety constraints,
designing appropriate barrier functions for
robotic systems operating in complex environments
is generally a hard problem. As a result, 
simple but rather conservative certificates are typically used in practice.
Additionally,
conflicts between the reference control laws and the barrier certificates
may introduce unwanted equilibria to the closed-loop system.
A method to address this problem is proposed in~\cite{18673424,46036224}
that learns control barrier functions from expert demonstrations. However, this method requires dense enough sampled data as well as the Liptchitz constants of
the system's dynamics and corresponding neural network controller that are hard to obtain in practice.

Compared to the control barrier function methods discussed above that can usually only ensure invariance of a conservative subset of the set of safe states, backwards reachability methods can instead compute the exact set of safe states which, similar to control barrier function methods, can then be rendered invariant by an appropriate design of controllers that take over when the system approaches the boundary of that safe set. However, exact computation of the safe set using backwards reachability methods generally comes at the expense of higher computational cost. Such a method is presented in~\cite{52351685} that employs Hamilton-Jacobi reachability to
compute the set of safe states for which control inputs exist
that prevent the system from violating safety specifications.
Then, a supervisory control policy is proposed that
feeds the output of a nominal neural network controller to the system
when the state lies inside the safe set
and
switches to a fail-safe optimal control law that can drive the system's state back to the safe set, when
the nominal controller drives the system outside the safe set.
This supervisory control law is extended in~\cite{33804135}
to address model uncertainty by
learning bounds on the unknown disturbances online and utilizing them
to derive probabilistic guarantees on safety.
In~\cite{73430503},
Hamilton-Jacobi reachability is also used to
generate safe and optimal trajectories that are
used as data to train convolutional neural networks
that yield waypoints which can be tracked robustly by
the robot equipped with a conventional feedback controller.

Common in the methods discussed above is that they generally apply reachability analysis on the open loop dynamics
without
considering the specific structure of
the reference neural network controller, owning to its complexity.
Reachability analysis of neural networks is an actively studied topic
and recent solutions have been proposed for
verification~%
\cite{75421458,88568115,77475528}
and
robust training~%
\cite{18463457,68156223}
alike.
These methods have been successfully adapted for
estimating the forward reachable set of dynamical systems
in feedback interconnection with feed-forward neural network controllers.
For example,
in~\cite{40700408,81450036},
Taylor models are used to compute an over-approximation of
the closed-loop system's reachable set by constructing
a polynomial approximation of the neural network controller.
An alternative approach to this problem is proposed in~\cite{24785682}
that translates activation functions into quadratic constraints
and
formulates a semi-definite program in order to compute
an over-approximation of forward reachable sets of the closed-loop system.
Over-approximations of the forward reachable set can also be computed using
interval subdivision methods as shown in~\cite{50517380}.
In this work,
the robot's state space is partitioned into
smaller cells with specified size
and
the controller's output set is approximated by
the union of each cell's output reachable set over-approximation.
To reduce the fineness of the partition and, effectively,
minimize computational cost,
heuristic methods have been developed in~\cite{32783264,03272408}.
%
%
Although the methods discussed above provide accurate over-approximations of
the reachable set of  neural network controllers,
they only address the verification problem of already trained controllers
and
do not consider safety specifications.
%
%
%

In this work,
we consider the problem of training a neural network controller
in feedback interconnection with a polygonal shaped robot,
able to translate and rotate within a compact, obstacle-occluded workspace,
such that
the closed-loop system safely reaches a specified goal configuration
and
possible safety violations are explicitly bounded.
Specifically,
given a dataset of state and input pairs
sampled from safe robot trajectories,
we first train a feed-forward neural network controller so that
the closed-loop system fits the vector field implicitly defined by the data.
Then,
we propose a new supervised learning method to iteratively re-train
the neural network controller so that
safety specifications are satisfied,
namely,
the robot is steered away from states that can result in
collisions with the obstacles.
To do so,
at each re-training iteration,
we compute an over-approximation of
the robot's forward reachable set under the current trained controller
using the Interval Bound Propagation (IBP) method~\cite{50517380}
and
update the loss function with appropriate penalty terms used to minimize
the overlap between
this over-approximation of the forward reachable set
and
the set of unsafe states.
Since the robot configuration space and, therefore, the safe set of states
needed to compute the robot's forward reachable set,
are difficult to obtain explicitly due to the arbitrary geometry of
the robot and the workspace,
we use the adaptive partitioning method proposed in~\cite{21101558} to
obtain a simpler over-approximation of this safe set.
Particularly,
we construct a cover of the set of feasible robot states
consisting of rectangular cells that
either
contain only safe configurations
or
intersect with the boundary of the safe set and, thus, contain
both safe and unsafe states.
Then,
we compute suitable under- and over-approximations of the robot's footprint
which we use to determine
whether a given cell contains only safe states or not.
Cells that lie on the boundary of the safe set and, therefore,
contain unsafe states, are recursively subdivided
until a tight enough cover is obtained.
Since the quality of the over-approximation of the forward reachable set
depends on the accuracy of the over-approximation of the safe set
and
on the parameters of the neural network controller,
which are updated with every re-training iteration,
we also propose a method to
refine the cell partition of the safe set at each iteration
by subdividing and/or merging cells based on the overlap between
the over-approximation of their forward reachable sets
and
the set of unsafe states.
This way,
we can control the tradeoff between
computational complexity and tightness of the safety violation bounds.
%
%
Finally,
we provide a simulation study that verifies the efficacy of
the proposed scheme.

To the best of our knowledge,
this is the first safe learning framework for
neural network controllers that (i) provides numerical worst-case bounds on safety violation
over the whole configuration space,
defined by the overlap between
the over-approximation of the forward reachable set of the closed-loop system
and
the set of unsafe states,
and (ii)
controls the tradeoff between computational complexity and tightness
of these bounds.
Compared to the methods in \cite{01373512,18673424,52351685} that employ control barrier functions or Hamilton-Jacobi reachability to design fail-safe projection operators or supervisory control policies, respectively, that can be wrapped around pre-trained nominal neural network controllers that are possibly unsafe due to, e.g., insufficient data during training, here we directly train neural network controllers with safety specifications in mind. On the other hand, unlike the methods in \cite{73430503,74521314,01373512} that also directly train safe neural network controllers using sufficiently many safe-by-design training samples, here safety of the closed-loop system does not depend on data points, but on  safety violation bounds defined over the whole configuration space that enter explicitly the loss function as penalty terms during training. As a result, when our method fails to guarantee safety of the closed-loop system in the whole configuration space, it does so with {\em grace} by also providing safety violation bounds whose size can be controlled. Such guarantees on the closed-loop performance cannot be obtained using the methods in \cite{73430503,74521314,01373512} that critically depend on the density of sampling.
While the heuristic methods in~\cite{32783264,03272408}
can obtain tight bounds on the controller outputs used to
verify the robustness of the controlled system,
they do so by considering the reachable set of the controller
instead of the closed-loop dynamics.
The key idea in this work that allows us to control tightness of
the safety violation bounds as a function of the computational cost is that
the partition of the safe set into cells is guided by
the overlap between the over-approximation of
the forward reachable set of the closed-loop system
and
the set of unsafe states, which itself measures safety violation.
Therefore,
cells located far from the boundary of the safe set
are not unnecessarily subdivided given that
they do not require precise control input bounds to be labeled as safe.
As a result,
the partition of the safe set computed by our method consists of
noticeably fewer cells compared to the number of cells returned by
brute-force subdividing the safe set to
achieve the same tightness of safety violation bounds.
%
%
The effect of a smaller partition is
fewer reachable set evaluations required at each iteration of
the training process; a unique feature of the proposed method. Our approach is inspired by the methods proposed in \cite{68156223,18463457} that
use reachability analysis to train robust neural networks for regression
and classification. Here, we extend these subdivision-based reachability
analysis algorithms to learn safe neural
network controllers. 

Perhaps most closely related to the method proposed here is the approach in \cite{60001271} that also trains provably safe neural network controllers for robot navigation. Specifically, given an arbitrary fixed partition of the robot’s state space and the controller’s parameter space, \cite{60001271} employs reachability analysis to build a transition graph and identify regions of control parameters (i.e., weights and biases) for which the closed-loop system is guaranteed to be safe. Then, a distinct ReLU neural network controller is trained for every state space cell using standard learning methods augmented by a projection operator that restricts the trained neural network weights to the regions found to be safe. Compared to \cite{60001271}, the method proposed here is not restricted to ReLU neural networks and can accommodate any class of strictly increasing continuous activation functions. Moreover, the method proposed here applies to non-point robots and returns smooth trajectories that respect the robot dynamics. To the contrary, \cite{60001271} applies to point robots and returns trajectories that are non-smooth due to possible discontinuities when the ReLU neural network controllers are stitched together at consecutive state space cells. Finally, feasibility of the control synthesis problem in \cite{60001271} strongly relies on the availability of sufficient data needed to train neural network weights that belong to the regions found to be safe. Instead, the method proposed here removes such strong assumptions on data that are hard to satisfy in practice and instead allows for graceful safety violations, of a bounded magnitude that can be spatially controlled.

We organize the paper as follows.
%
In \autoref{sec/problem},
we formulate the problem under consideration
while
in \autoref{sec/method}
we develop the proposed adaptive subdivision method that relies on reachability analysis techniques to partition the safe set into cells 
and
present the proposed safety-aware learning method.
Finally,
in \autoref{sec/sims}
we conclude this work by presenting results corroborating the
efficacy of our scheme.

\textit{Notation:}
We will refer to an interval of the set of real numbers \( \s{set/num/real} \)
as a simple interval,
whereas the Cartesian product of one or more simple intervals will be
referred to as a composite interval.
An \(n\)-dimensional interval is a composite interval of the form
\(
[ x_{l,1}, x_{u,1} ]
\times
[ x_{l,2}, x_{u,2} ]
\times
\ldots
\times
[ x_{l,n}, x_{u,n} ]
\).
Given
a compact set \( \mathcal{S} \in \s{set/num/real}^{n} \)
and
a vector
\( \delta = [\delta_1, \delta_2, \ldots, \delta_n] \in \s{set/num/real}^{n} \)
we shall use \( \s{set/rel/vol/scaled}{\mathcal{S}}{\delta} \) to denote
the volume of \( \mathcal{S} \) after scaling the \(i\)-th dimension by
the \( \delta_{i} \), i.e.,
\(
\s{set/rel/vol/scaled}{\mathcal{S}}{\delta} =
\prod_{i = 1}^{n}{
  \delta_i \cdot (x_{u,i} - x_{l,i})
}
\).
For brevity,
we shall write \( \s{set/rel/vol}{\mathcal{S}} \)
to denote
\( \s{set/rel/vol/scaled}{\mathcal{S}}{\delta} \)
when \( \delta = [1, 1, \ldots, 1] \). Given a set, $\mathcal{S}$, $\s{set/Realsets}(\mathcal{S})$ is the set of all its compact subsets.  We will denote $\mathcal{C}_{r}=[-r,r]^3$ and $\mathcal{B}_r = \{x\in\s{set/num/real}| ||x||\leq r\}$ The zero element in any vector space will be denoted $\mathrm{0}$.

\section{Problem Formulation}%
\label{sec/problem}

\begin{figure}
  \centering
  \includegraphics[width=0.9\linewidth]{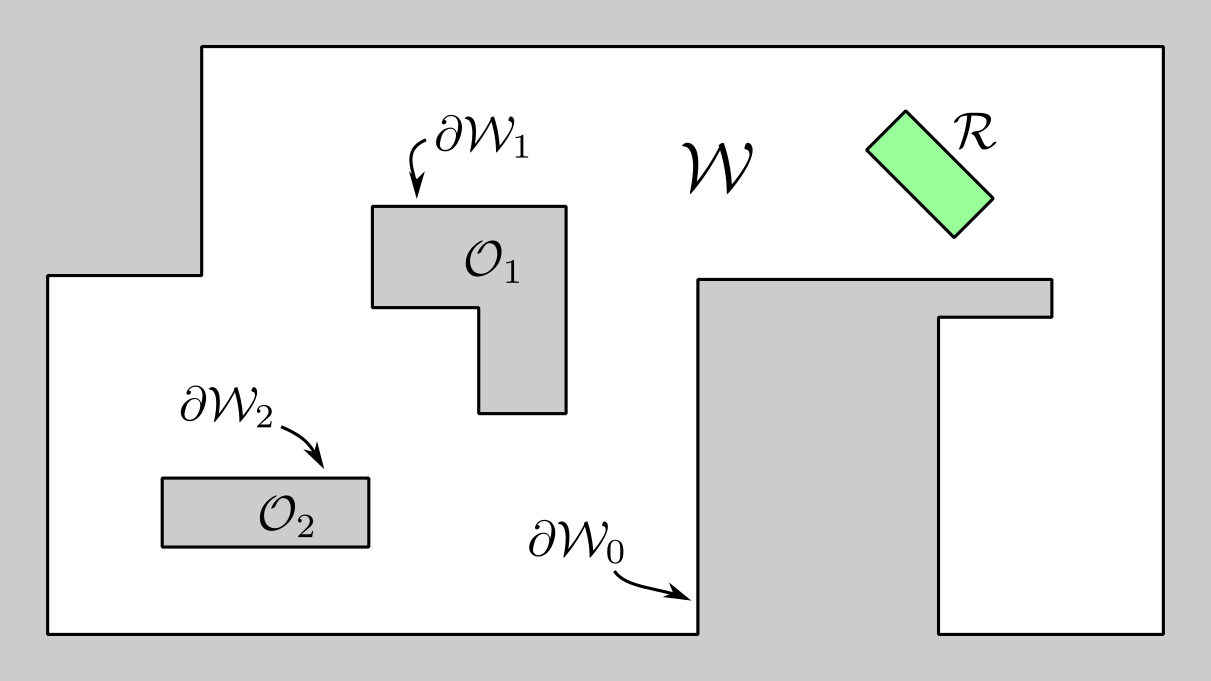}
  \caption{
    Robot \( \s{rob/body} \) operating within
    a workspace \( \s{ws} \) occupied by two inner obstacles.
  }%
  \label{fig/ws/example}
\end{figure}

We consider
a polygonal shaped robot \( \s{rob/body} \) operating within
a compact, static workspace \(\s{ws} \s{set/rel/sub} \s{set/num/real}^{2}\)
defined by
an outer boundary \( \s{ws/bnd/out} \)
and
\( \s{ws/obst/amount} \) inner boundaries
\(
\s{ws/bnd}{i},
i \in \s{ws/obst/idxset} \triangleq
\s{set/enum}{1, 2, \ldots, \s{ws/obst/amount}}
\),
corresponding to 
a set of disjoint, fixed inner obstacles
\( \s{ws/obst}{i} \),
as seen in \autoref{fig/ws/example}.
%
%
We assume that
the boundary \( \s{rob/body/bnd} \) of the robot's body
and
the boundaries \( \s{ws/bnd/out}, \s{ws/bnd/ins}{i} \) of the workspace
are polygonal Jordan curves.
Let
\(\s{frame/ws}\) and \(\s{frame/rob}\) be coordinate frames
arbitrarily embedded in the workspace and on the robot, respectively,
and
let
\(\s{rob/cfg} = \s{lrsb}{\s{rob/pos/x}, \s{rob/pos/y}, \s{rob/angle}}^T\)
denote the configuration of the robot on the plane,
specifying the relative position
\(
\s{lrsb}{\s{rob/pos/x}, \s{rob/pos/y}}^T \in \s{set/num/real}^2
\)
and
the orientation
\(
\s{rob/orient/at}{\s{rob/angle}} =
{[\cos{\s{rob/angle}}, \sin{\s{rob/angle}}]}^T
\),
\(\s{rob/angle} \in [ 0, 2\pi )\),
of \(\s{frame/rob}\) with respect to \(\s{frame/ws}\).
We assume that
the robot is able to translate and rotate subject to
the following discrete-time non-linear dynamics
\begin{equation}
  \label{eq/sys}
  \s{rob/state/at}{k+1} =
  \s{rob/dynamics}{\s{rob/state/at}{k}}{\s{rob/input/at}{k}}
  ,
\end{equation}
where
\(
\s{rob/state} = \s{lrsb}{\s{rob/cfg}^T, \s{rob/state/misc}^T }^T
\in \s{rob/ss} \s{set/rel/sub} \s{set/num/real}^{\s{rob/state/dim}}
\),
and
\( \s{rob/input} \in \s{set/num/real}^{\s{rob/input/dim}} \)
denote the robot's state and control input,
respectively, and $q$ denotes miscellaneous
robot states, e.g., linear and angular velocities, accelerations, etc. We define
\(
\s{rob/ss} \triangleq
\s{rob/ss/cfg} \times \s{rob/ss/misc}
\),
where
\( \s{rob/ss/cfg} = \s{ws} \times \s{set/num/circle} \),
and
\(
\s{rob/ss/misc}
\s{set/rel/sub}
\s{set/num/real}^{\s{rob/state/dim}-3}
\)
is a \( (\s{rob/state/dim}-3) \)-dimensional interval of
miscellaneous safe robot states, e.g.,
allowed bounds on the robot's velocities.
%
%
The non-linear function
\(
\s{rob/f} :
\s{set/num/real}^{\s{rob/state/dim}} \mapsto \s{set/num/real}^{\s{rob/state/dim}}
\)
is assumed known and continuously differentiable.
Furthermore,
we assume that
a function \( \s{rob/f/oa} \) is known which maps
composite intervals \( \mathcal{A} \) and \( \mathcal{B} \)
of \( \s{rob/ss} \) and \( \s{rob/is} \), respectively,
to composite intervals of \( \s{rob/ss} \) such that
\begin{equation}%
  \label{eq/f/oa}
  \s{rob/f/at}{\s{rob/state}}{\s{rob/input}} \in
  \s{rob/f/oa/at}{\mathcal{A}}{\mathcal{B}}
  , ~~~
  \forall \s{rob/state}, \s{rob/input} \in
  \mathcal{A} \times \mathcal{B}
  .
\end{equation}

In order to steer the robot to
a desired feasible state \( \s{rob/state/goal} \),
we equip the robot with a feed-forward multi-layer neural network controller
\(
\s{nn}:
\s{set/num/real}^{\s{nn/in/dim}} \mapsto \s{set/num/real}^{\s{nn/out/dim}}
\)
that
consists of \( \s{nn/layers/amount} \) fully connected layers, i.e.,
\begin{equation}%
  \label{eq/nn/def}
  \s{rob/input} =
  \s{nn}(\s{rob/state}) =
  \s{nn/layer/map/at}{\s{nn/layers/amount}}{
    \s{nn/layer/map/at}{\s{nn/layers/amount}-1}{
      \ldots
      \s{nn/layer/map/at}{1}{\s{rob/state}}
      \ldots
    }
  }
  ,
\end{equation}%
\begin{equation}%
  \label{eq/nn/def/layer}
  \s{nn/layer/map/at}{i}{x} =
  \s{nn/layer/afun/at}{i}{
    \s{nn/layer/weight}{i} \cdot x + \s{nn/layer/bias}{i}}
  ,~~~
  \forall i \in \s{nn/layers/idxset}
  ,
\end{equation}
where
\( \s{nn/layer/weight}{i}, \s{nn/layer/bias}{i}, \s{nn/layer/afun}{i} \)
denote the weight matrix, bias vector and activation function
of the \( i \)-th layer,
for all \( i \in \s{nn/layers/idxset} \).
Specifically,
the \( i \)-th layer of \( \s{nn} \) consists of
\( \s{nn/layer/size}{i} \) neurons, i.e.,
\(
\s{nn/layer/weight}{i} \in
\s{set/num/real}^{\s{nn/layer/size}{i} \times \s{nn/layer/size}{i-1} }
\),
\(
\s{nn/layer/bias}{i} \in
\s{set/num/real}^{\s{nn/layer/size}{i} \times 1 }
\),
\(
\s{nn/layer/afun}{i} \in \s{set/num/real}^{\s{nn/layer/size}{i} \times 1}
\),
for all \( i \in \s{nn/layers/idxset} \)
where
\( \s{nn/layer/size}{0} = \s{nn/in/dim} \)
and
\( \s{nn/layer/size}{\s{nn/layers/amount}} = \s{nn/out/dim} \).

To find values for
\( \s{nn/layer/weight}{i}, \s{nn/layer/bias}{i}, i \in \s{nn/layers/idxset} \)
such that
the controller \( \s{nn} \)
can steer the robot to the desired state \( \s{rob/state/goal} \),
we assume that we are given 
a set \( \s{data/set} \) consisting of points
\(
(\s{data/point/ith/state}{i}, \s{data/point/ith/input}{i}) \in
\s{rob/ss} \times \s{rob/is}
\)
sampled from robot trajectories
beginning at different random states 
\( \s{rob/state} \in \s{rob/ss} \)
and
terminating at \( \s{rob/state/goal} \).
Additionally,
we assume that the given trajectories are feasible in the sense that
they consist entirely of states that are safe.
A state
\( \s{rob/state} = \s{lrsb}{\s{rob/state/cfg}^T, \s{rob/state/misc}^T}^T \)
is said to be safe if
\( \s{rob/state/misc} \in \s{rob/ss/misc} \)
and
the robot at the corresponding configuration does not overlap with
any of the static obstacles, i.e.,
\(
\s{rob/body/at}{\s{rob/state/cfg}} \s{set/rel/sub} \s{ws}
\)
where
\( \s{rob/body/at}{\s{rob/cfg}} \),
denotes
the robot's footprint when
it is placed at \( \s{lrsb}{\s{rob/pos/x}, \s{rob/pos/y}}^T \) with
orientation \( \s{rob/angle} \).
If the parameter space of the network \( \s{nn} \)
and
the dataset \( \s{data/set} \)
are large enough,
parameters \( \s{nn/weight}, \s{nn/bias} \) can be typically found by solving
the optimization problem
\begin{equation}%
  \label{eq/optim/base}
  \begin{aligned}
    & \underset{
      \s{nn/weight}, \s{nn/bias}
    }{\text{minimize}}
    & & \s{ml/base/costfun/at}{\s{nn/weight}, \s{nn/bias}} \\
  \end{aligned}
  ,
\end{equation}
where
\begin{equation}
  \label{eq/optim/base/costfun}
  \s{ml/base/costfun/at}{\s{nn/weight}, \s{nn/bias}} =
  \s{ml/base/coef/err}
  \cdot
  \sum_{\s{lrp}{\s{data/point/state}, \s{data/point/input}} \in \s{data/set}}{
    \s{lrp}{
      \s{data/point/input} -
      \s{nn/at}{\s{data/point/state}; \s{nn/weight}, \s{nn/bias}}
    }^{2}
  }
  +
  \s{ml/base/coef/reg}
  \cdot
  \s{ml/base/reg/at}{\s{nn/weight}, \s{nn/bias}}
  .
\end{equation}
In the loss function~\eqref{eq/optim/base/costfun},
\(
\s{nn/weight} = \s{lrsb}{
  \s{nn/layer/weight}{1},
  \s{nn/layer/weight}{2},
  \ldots,
  \s{nn/layer/weight}{\s{nn/layers/amount}}
}
\)
,
\(
\s{nn/bias} = \s{lrsb}{
  \s{nn/layer/bias}{1},
  \s{nn/layer/bias}{2},
  \ldots,
  \s{nn/layer/bias}{\s{nn/layers/amount}}
}
\)
,
\( \s{ml/base/reg/at}{\cdot} \)
is a regularization term
and
\( \s{ml/base/coef/err} \),
\( \s{ml/base/coef/reg} \)
are positive constants.
%
%
We remark that
the controller obtained from the solution of
the problem \eqref{eq/optim/base}-\eqref{eq/optim/base/costfun}
is expected to be safe only around the trajectories in
the training dataset \( \s{data/set} \).
This behavior is generally not desirable.
Instead,
what is desired is that
the robot dynamics~\eqref{eq/sys}
under the control law~\eqref{eq/nn/def}
ensure that the safe set of states \( \s{rob/ss/safe} \) remains invariant,
where \( \s{rob/ss/safe} \) consists of
all the states
\( \s{lrp}{\s{rob/state/cfg}, \s{rob/state/misc}} \in \s{rob/ss} \)
such that \( \s{rob/body/at}{\s{rob/state/cfg}} \s{set/rel/sub} \s{ws} \).
%
%
Therefore,
in this paper we consider the following problem.
\begin{problem}%
  \label{stmt/prob}
  Given a static workspace \( \s{ws} \),
  a polygonal robot \( \s{rob/body} \) subject to dynamics \( \s{rob/f} \),
  a safe set of miscellaneous robot states \( \s{rob/ss/misc} \),
  and
  a dataset \( \s{data/set} \),
  train a neural network controller \( \s{nn} \) so that
  the closed-loop trajectories fit the data in the set \( \s{data/set} \)
  and
  the safe set \( \s{rob/ss/safe} \) either remains invariant or 
  possible safety violations are explicitly bounded.
\end{problem}

\section{Methodology}%
\label{sec/method}

In order to address \autoref{stmt/prob},
we 
first employ standard learning methods to solve
the optimization problem~\eqref{eq/optim/base}
and
obtain initial values for
the parameters \( \s{nn/weight} \) and \( \s{nn/bias} \).
As discussed before,
the controller \( \s{nn} \) obtained at this stage
is expected to be safe only
around the points in \( \s{data/set} \),
assuming that
they have been sampled from safe trajectories and
the network
fits the dataset adequately.
Next,
we employ the subdivision method presented in
\autoref{sec/method/feas}
to obtain a partition $\mathcal{P}$ of the feasible space into cells that
provide a tight over-approximation \( \s{rob/ss/safe/oa} \) of
the robot's safe state space \( \s{rob/ss/safe} \).
Using the over-approximation of the safe set \( \s{rob/ss/safe} \) as
a set of initial robot states, in \autoref{sec/method/reach},
we compute
an over-approximation \( \s{rob/ss/frs/crit/oa} \) of
the forward reachable set \( \s{rob/ss/frs/crit} \) of
the closed loop system under the neural network controller \( \s{nn} \).
Since the accuracy of the over-approximation \( \s{rob/ss/frs/crit/oa} \)
depends on the partition \( \s{rob/ss/partition} \) of
the over-approximation of the safe set \( \s{rob/ss/safe/oa} \), we also propose a method to refine the partition \( \s{partition} \)
in order to improve the accuracy of
the over-approximation \( \s{rob/ss/frs/crit/oa} \).
Finally,
in \autoref{sec/method/optim},
we use the overlap between
the over-approximation of
the forward reachable set \( \s{rob/ss/frs/crit/oa} \)
and
the set of unsafe states to design penalty terms in
problem~\eqref{eq/optim/base}
that explicitly capture safety specifications.
As the parameters \( \s{nn/weight}, \s{nn/bias} \) of
the neural network \( \s{nn} \) get updated,
so does the shape of the forward reachable set \( \s{rob/ss/frs/crit/oa} \),
which implies that the initial partition \( \s{rob/ss/partition} \) of
the over-approximation of the safe set \( \s{rob/ss/safe/oa} \)
does no longer accurately approximate \( \s{rob/ss/frs/crit} \).
For this reason,
we repeat the partition refinement and training steps
proposed in \autoref{sec/method/reach} and \autoref{sec/method/optim} for
a sufficiently large number of epochs \( \s{ml/aug/epochs/amount} \).
The procedure described above is outlined in \autoref{alg/main}.

\IncMargin{1em}
\begin{algorithm2e}
    \caption{Safety-aware controller design.\label{alg/main}}%
     \Fn{$\{\s{nn/weight},\ \s{nn/bias}\}=$\Call{Train}{$\s{rob/ss/partition/wthresh},\ \s{rob/ss/refine/vphresh},\ \s{rob/ss/refine/vqhresh},\ \s{ml/aug/epochs/amount} $}
    }{%
        $ \s{nn/weight}, \s{nn/bias} \gets $ Solve problem~\eqref{eq/optim/base}\;
        $\s{rob/ss/partition}\gets$\Call{BuildPrtn}{$ \s{ws},\s{rob/body},\s{rob/ss/partition/wthresh} $} \tcp*{Use \autoref{alg/ss/partition}}
        \For{ $i$ in $ 1, 2, \ldots, \s{ml/aug/epochs/amount} $
        }{
            $ \s{rob/ss/partition} \gets $\Call{RefinePrtn}{$\s{rob/ss/partition},\s{rob/ss/partition/wthresh},\s{rob/ss/refine/vphresh},\s{rob/ss/refine/vqhresh};\s{ws},\s{rob/body},\s{rob/ss/misc},f$}\tcp*{Use \autoref{alg/ss/refinement}}
            $ \s{nn/weight}, \s{nn/bias} \gets $ Solve problem~\eqref{eq/optim/aug}\tcp*{Update $ \s{nn/weight}, \s{nn/bias} $}
        }
    }
\end{algorithm2e}
\DecMargin{1em}


\subsection{Over-Approximation of the Safe State Space}%
\label{sec/method/feas}


To obtain
a tight over-approximation \( \s{rob/ss/safe/oa} \) of
the robot's safe state space,
that will be used later to compute
the forward reachable set of the closed-loop system
as well as its overlap with the set of unsafe states,
we adaptively partition the safe state space \( \s{rob/ss/safe} \)
into cells using the adaptive subdivision method proposed in~\cite{21101558}.
Specifically,
starting with a composite interval enclosing
the set of feasible states \( \s{rob/ss/safe} \),
we recursively subdivide this interval into subcells based on
whether appropriately constructed under- and over-approximations of
the robot's footprint intersect with the workspace's boundary.
The key idea is that cells for which
the under-approximation (resp. over-approximation) of the robot's footprint
overlaps (resp. does not overlap) with the complement of \( \s{ws} \)
contain only unsafe (resp. safe) states and
subdividing them any further will not improve the accuracy of the partition
whereas
cells which contain both safe and unsafe states should be further
subdivided into subcells as
they reside closer to the boundary of \( \s{rob/ss/safe} \).
%
%
This procedure is repeated until
cells that constitute the partition of safe state space either
contain only safe states or
intersect with the boundary of \( \s{rob/ss/safe} \) and
their size is below a user-specified threshold.

We begin by recalling that
the set \( \s{rob/ss/safe} \) is defined as
\( \s{rob/ss/cfg/safe} \times \s{rob/ss/misc} \)
where
\( \s{rob/ss/misc} \) is
a composite interval
and
\( \s{rob/ss/cfg/safe} \)
is defined as the largest subset of \( \s{rob/ss/cfg} \)
such that
\( \s{rob/body/at}{\s{rob/state}} \s{set/rel/subeq} \s{ws} \)
for all \( \s{rob/state} \in \s{rob/ss/cfg/safe} \).
Thus,
in order to compute \( \s{rob/ss/safe/oa} \),
we need to find
a valid over-approximation
\( \s{rob/ss/cfg/safe/oa} \) of \( \s{rob/ss/cfg/safe} \).
Additionally,
we require that
the over-approximation \( \s{rob/ss/safe/oa} \) is defined by
the union of a finite number of composite intervals, i.e., cells.
This construction is necessary to obtain the forward reachable set of
the closed loop system using the method proposed in \autoref{sec/method/reach}.
%
%
We now consider a composite interval \( \s{rob/ss/cell} \) in
the robot's state space \( \s{rob/ss} \)
which we shall refer to as a state space cell.
Each state space cell \( \s{rob/ss/cell} \) can be written as
\( \s{rob/ss/cfg/cell} \times \s{rob/ss/misc/cell} \),
where
\( \s{rob/ss/cfg/cell} \in \s{rob/ss/cfg} \)
denotes a configuration space cell, i.e.,
a set of robot's positions and orientations.
%
We notice that
if \( \s{rob/body/at}{\s{rob/cfg}} \s{set/rel/subeq} \s{ws} \)
for all \( \s{rob/cfg} \in \s{rob/ss/cfg/cell} \),
then the cell \( \s{rob/ss/cfg/cell} \) consists entirely of
safe robot configurations and thus must lie entirely
inside \( \s{rob/ss/cfg/safe} \).
On the contrary,
if
\(
\s{rob/body/at}{\s{rob/cfg}} \s{set/rel/isect} \s{ws} \neq
\s{rob/body/at}{\s{rob/cfg}}
\)
for all \( \s{rob/cfg} \in \s{rob/ss/cfg/cell} \),
then the cell \( \s{rob/ss/cfg/cell} \) consists entirely of
unsafe robot configurations and thus must lie
outside \( \s{rob/ss/cfg/safe} \).
%
%
Since checking the above conditions to classify cells as safe or unsafe
is not easy due to the complex shape of
\( \s{rob/ss/safe} \) and the robot,
we instead employ for this purpose over- and under-approximations of
the robot's footprint,  
\( \s{rob/fpoa/at}{\s{rob/ss/cfg/cell}} \)
and
\( \s{rob/fpua/at}{\s{rob/ss/cfg/cell}} \), respectively, 
associated with
the configuration space cell \( \s{rob/ss/cfg/cell} \)
(see \autoref{fig/footprints})
that satisfy
\begin{equation}
  \begin{aligned}
    \s{rob/fpoa/at}{\s{rob/ss/cfg/cell}}
    &
    \s{set/rel/supeq}
    \s{set/rel/union/multi}{\s{rob/cfg} \in \s{rob/ss/cfg/cell}}{}{
      \s{rob/body/at}{\s{rob/cfg}}
    }
    \text{, and }
    \s{rob/fpua/at}{\s{rob/ss/cfg/cell}}
    &
    \s{set/rel/subeq}
    \s{set/rel/isect/multi}{\s{rob/cfg} \in \s{rob/ss/cfg/cell}}{}{
      \s{rob/body/at}{\s{rob/cfg}}
    }
    .
  \end{aligned}\label{eq:robotbounds}
\end{equation}
We remark that
such over- and under-approximations of the robot's footprint
can be easily computed for
box shaped cells \( \s{rob/ss/cfg/cell} \) using techniques such as
the \textit{Swept Area Method}~\cite{21101558}.
Additionally,
we notice that a cell \( \s{rob/ss/cfg/cell} \) for which
\( \s{rob/fpoa/at}{\s{rob/ss/cfg/cell}} \s{set/rel/subeq} \s{ws} \)
(resp.
\( \s{rob/fpua/at}{\s{rob/ss/cfg/cell}} \s{set/rel/not/subeq} \s{ws} \))
lies entirely inside (resp. outside) \( \s{rob/ss/cfg/safe} \).
We shall refer
to the first type of cells as safe
and
to the second as unsafe.
Cells not belonging to either of these two classes intersect with
the boundary of \( \s{rob/ss/cfg/safe} \) and will be labeled as mixed.
Our goal is to approximate \( \s{rob/ss/cfg/safe} \) by
the union of a finite number of safe cells but, in general,
the shape of the robot's configuration space
does not admit such representations.
As such,
we instead compute
a cover of \( \s{rob/ss/cfg/safe} \) consisting of
both safe and mixed cells,
where the size of the mixed cells should be kept as small as possible.
%
%

\begin{figure}
  \centering
  \includegraphics[width=0.6\linewidth]{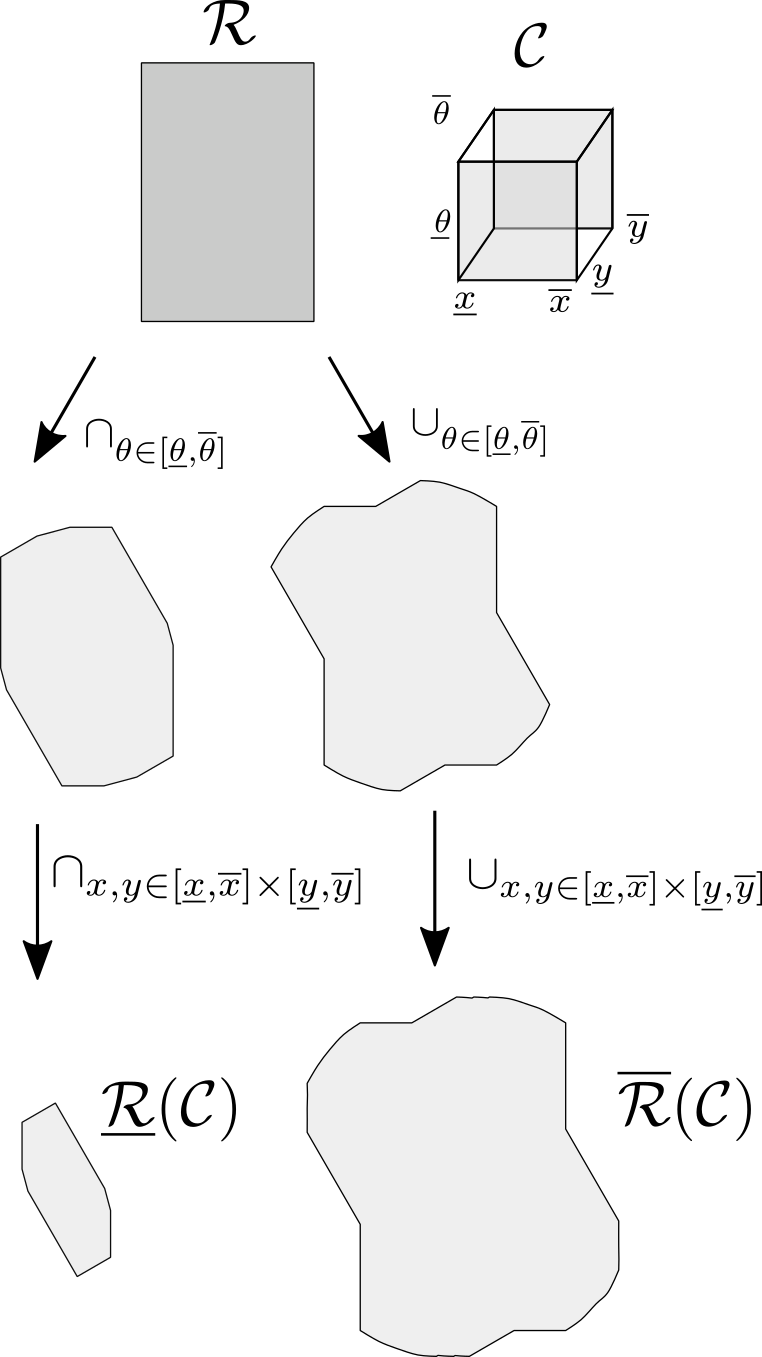}
  \caption{
    Under-approximation \( \s{rob/fpua/at}{\s{rob/ss/cell}} \) (left)
    and
    over-approximation \( \s{rob/fpoa/at}{\s{rob/ss/cell}} \) (right)
    of a rectangular shaped robot corresponding to
    the configuration space cell
    \(
    \s{rob/ss/cell} =
    [\underline{\s{rob/pos/x}}, \overline{\s{rob/pos/x}}] \times
    [\underline{\s{rob/pos/y}}, \overline{\s{rob/pos/y}}] \times
    [\underline{\s{rob/angle}}, \overline{\s{rob/angle}}]
    \).
  }%
  \label{fig/footprints}
\end{figure}

To do so,
we propose the adaptive subdivision method outlined in
\autoref{alg/ss/partition}.
Beginning with the state space cell defined by
the Cartesian product of
the axis-aligned bounding box of \( \s{ws} \)
and
the set \( \s{rob/ss/misc} \),
the proposed algorithm builds a partition \( \s{rob/ss/partition} \) of
the robot's state space by adaptively subdividing cells
that lie on its frontier \( \s{rob/ss/frontier} \).
Specifically,
at each iteration,
a cell \( (\s{rob/ss/cfg/cell}, \s{rob/ss/misc/cell}) \in \s{rob/ss/frontier} \)
is selected and if \( \s{rob/fpoa/at}{\s{rob/ss/cfg/cell}} \)
lies inside \( \s{ws} \),
then it is added to \( \s{rob/ss/partition} \). 
%
Instead,
if \( \s{rob/fpoa/at}{\s{rob/ss/cfg/cell}} \) overlaps with
the complement of \( \s{ws} \), then that cell gets discarded. 
%
If, at this point,
the cell cannot be labeled either safe or unsafe
and
its size is not smaller than
a user specified threshold \( \s{rob/ss/partition/wthresh} \),
then the cell gets subdivided and the new cells are added to
the frontier \( \s{rob/ss/frontier} \). 
%
Otherwise,
if the cell's size is smaller than the specified threshold
and
cannot be subdivided any more,
it is included in the partition \( \s{rob/ss/partition} \). 
%
Finally,
the desired over-approximation \( \s{rob/ss/safe/oa} \) can be obtained as
the union of the cells in the computed partition \( \s{rob/ss/partition} \).

In the following theorem we show that \autoref{alg/ss/partition} terminates in finite time and its worst case run-time is related to the size of the workspace and the reciprocal of the volume of the smallest allowable cell, as expected. Moreover, we show that the resulting cell partition $\s{rob/ss/partition}$ lies in the workspace and covers the robot's entire safe set, exceeding it in proportion to a user-selected tolerance. Finally, we bound the possible safety violations that can result from the operation of the robot  in $\s{rob/ss/partition}$. Specifically, we show that violations are proportional to the size of translation outside of $\s{rob/ss/safe}$, but rotations outside of $\s{rob/ss/safe}$ are amplified by the maximum distance from any point in the robot to its center of rotation.

\begin{theorem}\label{thm:InitialPartition}
Suppose a compact set $\s{ws} \s{set/rel/sub} \s{set/num/real}^{2}$, a mapping $\s{rob/body}:\s{set/num/real}^{2}\times[0,2\pi) \mapsto \s{set/Realsets}(\s{set/num/real}^{2})$, and $\s{rob/ss/partition/wthresh}>0$. Let also $\s{rob/fpoa},\s{rob/fpua}:\s{set/Realsets}(\s{set/num/real}^{2}\times[0,2\pi))\mapsto\s{set/Realsets}(\s{set/num/real}^{2})$ be mappings satisfying \autoref{eq:robotbounds}. Let $\s{rob/ss/frontier}_0$ be the initialization of $\s{rob/ss/frontier}$. Then \autoref{alg/ss/partition} terminates in finite time after at most $2\s{rob/ss/partition/wthresh}^{-3}\s{set/rel/vol}{\s{rob/ss/frontier}_0} $ repetitions of the while loop. It's output, $\s{rob/ss/partition}$, is a collection of cells in $\s{ws}$ satisfying  $\s{rob/ss/safe} \subseteq \cup_{\s{rob/ss/cell}\in \s{rob/ss/partition}}\s{rob/ss/cell}  \subseteq  \s{rob/ss/safe}\oplus (\mathcal{C}_{\s{rob/ss/partition/wthresh}}\times \mathrm{0})$. 

Suppose further that $\s{rob/body}([x,y,\theta]^\trans)= R(\theta)\mathrm{R}+[x,y]^\trans$, where $R(\theta)$ is a rotation in $\s{set/num/real}^{2}$, and $\mathrm{R}\subseteq\mathcal{B}_r$ for some $r>0$. In this case, $ \cup_{\s{rob/ss/cell}\in \s{rob/ss/partition}}\s{rob/ss/cell}  \cup_{p\in \s{rob/ss/cfg/cell}} \s{rob/body}(p) \subseteq \s{ws}\oplus\mathcal{B}_{v}$ where $v=2r^2(1-\cos(\s{rob/ss/partition/wthresh}))+2\sqrt{2}\s{rob/ss/partition/wthresh}$.
\end{theorem}

\begin{proof}
To begin, observe that all operations in \autoref{alg/ss/partition} are well-defined and can be computed in finite time due to the compactness assumption on $\s{ws}$ and the outputs of $\s{rob/body}$. All relevant set widths are finite due to this assumption.

To bound the while loop's repetitions, observe that each repetition either: removes a cell from $\s{rob/ss/frontier}$, reducing $\s{set/rel/vol}{\s{rob/ss/frontier}}$ by at least the minimum cell volume, $\s{rob/ss/partition/wthresh}^{3}$, or replaces one cell in $\s{rob/ss/frontier}$ with two covering the same volume, leaving $\s{set/rel/vol}{\s{rob/ss/frontier}}$ unchanged, but increasing the number of cells by one. The maximum number of cells in $\s{rob/ss/frontier}$ is bounded above by $k = \s{rob/ss/partition/wthresh}^{-3}\s{set/rel/vol}{\s{rob/ss/frontier}_0}$ because when $k$ cells are in $\s{rob/ss/frontier}$, their the total volume is at least $k\s{rob/ss/partition/wthresh}^{3}$, but $\s{set/rel/vol}{\s{rob/ss/frontier}}$ can never grow. Hence, at most $k$ repetitions can add cells to $\s{rob/ss/frontier}$, and $k$ can remove cells.

To verify that $\s{rob/ss/safe} \subseteq \cup_{\s{rob/ss/cell}\in \s{rob/ss/partition}}\s{rob/ss/cell}$, suppose that it fails to hold. Then there exists $z\in \s{rob/ss/cfg/safe}$ where $z\notin \union_{\s{rob/ss/cell}\in\s{rob/ss/partition}}\s{rob/ss/cfg/cell}$. However, by design $z\in\union_{\s{rob/ss/cfg/cell}\in\s{rob/ss/frontier}_0}\s{rob/ss/cfg/cell}$. This can only happen if a cell containing $z$, $\s{rob/ss/cfg/cell}$, was removed because $ \s{rob/fpua/at}{\s{rob/ss/cfg/cell}} \s{set/rel/not/subeq} \s{ws} $. This implies that for all $p\in\s{rob/ss/cfg/cell}$, $\s{rob/body}(p)\not\subseteq\s{ws}$. In particular, this must hold for $p=z$, contradicting the assumption that $z\in \s{rob/ss/cfg/safe}$.

To verify $\cup_{\s{rob/ss/cell}\in \s{rob/ss/partition}}\s{rob/ss/cell} \subseteq \s{rob/ss/safe}\oplus (\mathcal{C}_{\s{rob/ss/partition/wthresh}}\times \mathrm{0})$, note that cells can only be added to $\s{rob/ss/partition}$ in two ways. First, $\s{rob/ss/cell}$ can be added if $\s{rob/fpoa/at}{\s{rob/ss/cfg/cell}} \s{set/rel/subeq} \s{ws}$, which implies $\s{rob/ss/cfg/cell}\subseteq\s{rob/ss/cfg/safe}$, and hence $\s{rob/ss/cell}\subseteq\s{rob/ss/safe}\subset  \s{rob/ss/safe}\oplus (\mathcal{C}_{\s{rob/ss/partition/wthresh}}\times \mathrm{0})$. Otherwise, $\s{rob/ss/cell}$ can only be added if $ \s{rob/fpua/at}{\s{rob/ss/cfg/cell}} \s{set/rel/not/subeq} \s{ws} $ and \Call{GetMaxWidth}{$ \s{rob/ss/cfg/cell} $} $ > \s{rob/ss/partition/wthresh}$ fail to hold. These respectively imply that there exists $z\in \s{rob/ss/cfg/cell}\cap \s{rob/ss/cfg/safe}$ and that $\s{rob/ss/cfg/cell}\subseteq z \oplus \mathcal{C}_{\s{rob/ss/partition/wthresh}} \subseteq \s{rob/ss/cfg/safe}\oplus \mathcal{C}_{\s{rob/ss/partition/wthresh}}$, which means $\s{rob/ss/cell} \subseteq \s{rob/ss/safe}\oplus (\mathcal{C}_{\s{rob/ss/partition/wthresh}}\times \mathrm{0})$.

For the final assertion, suppose $p\in \s{rob/ss/cfg/cell}$ for some $\s{rob/ss/cell} \in \s{rob/ss/partition}$. From the prior assertions, $p=[x,\theta]^\trans+[\delta x,\delta\theta]^\trans$, where $[x,\theta]^\trans\in \s{rob/ss/cfg/safe}$ and $[\delta x,\delta\theta]^\trans\in\mathcal{C}_{\s{rob/ss/partition/wthresh}}$. Therefore $\s{rob/body}(p) = R(\delta\theta)R(\theta)\mathrm{R}+x+\delta x$. Using the Cosine Rule and recalling that rotations preserve norms, $||x-R(\delta\theta)x||=2||x||^2(1-\cos(\delta\theta))\leq 2r^2(1-\cos(\s{rob/ss/partition/wthresh}))$ if $x\in R(\theta)\mathrm{R}$. Hence, $\s{rob/body}(p)\subseteq \s{rob/body}([x,\theta]^\trans)\oplus\mathcal{B}_{ 2r^2(1-\cos(\s{rob/ss/partition/wthresh}))}\oplus[-\s{rob/ss/partition/wthresh},\s{rob/ss/partition/wthresh}]^2$. Noting that $\mathcal{B}_{ 2r^2(1-\cos(\s{rob/ss/partition/wthresh}))}\oplus[-\s{rob/ss/partition/wthresh},\s{rob/ss/partition/wthresh}]^2\subseteq \mathrm{B}_{v}$ gives the desired result.
\end{proof}

\begin{remark}
Note that \autoref{alg/ss/partition} does not produce a mathematical partition of $\s{rob/ss/safe/oa}$ because the cells include their boundaries, where they may overlap. If one is needed, simple modifications can be made to the cell definition and subdivision algorithms to include boundaries in only one cell without fundamentally changing \autoref{alg/ss/partition}.
\end{remark}

\IncMargin{1em}
\begin{algorithm2e}
    \caption{Given a robot described by, $\s{rob/body}$, an allowable workspace, $\s{ws}$, and a tolerance, $\s{rob/ss/partition/wthresh}$, produce a cell collection, $\s{rob/ss/partition}$, covering $ \s{rob/ss/safe/oa}$ with bounded safety violations.\label{alg/ss/partition}}%
    \Fn{$\s{rob/ss/partition}=$\Call{BuildPrtn}{$ \s{ws},\s{rob/body},\s{rob/ss/partition/wthresh} $}}{
         $\s{rob/ss/partition},\ \s{rob/ss/frontier}\ 
         \gets 
         \s{set/enum}{},\ \left\{ \text{A bounded cell covering of }\s{ws}\right\}\times\s{set/num/circle}$ \;
         \While{$\s{rob/ss/frontier}$ is not empty}{
            $\s{rob/ss/cfg/cell}\ \gets \text{a cell in }\s{rob/ss/frontier}$ \;
            $\s{rob/ss/frontier} \ \gets\ \s{rob/ss/frontier}\setminus \s{rob/ss/cfg/cell}$\;
            \uIf{$ \s{rob/fpoa/at}{\s{rob/ss/cfg/cell}} \s{set/rel/subeq} \s{ws} $}{
                $ \s{rob/ss/partition} \gets \s{rob/ss/partition} \s{set/rel/union} \s{set/enum}{ \s{rob/ss/cfg/cell}\times \s{rob/ss/misc} } $
            }
            \ElseIf{$\s{rob/fpua/at}{\s{rob/ss/cfg/cell}} \s{set/rel/subeq} \s{ws}$}{
                \uIf{\Call{MaxWidth}{$\s{rob/ss/cfg/cell}$}$ < \s{rob/ss/partition/wthresh}$}{
                     $\s{rob/ss/frontier} \gets \s{rob/ss/frontier}\cup$\Call{SubdivideCell}{$\s{rob/ss/cfg/cell}$}
                }
                \Else{
                $ \s{rob/ss/partition} \gets \s{rob/ss/partition} \s{set/rel/union} \s{set/enum}{ \s{rob/ss/cfg/cell}\times \s{rob/ss/misc} } $
                }
            }   
        }
    }
\end{algorithm2e}
\DecMargin{1em}

\subsection{
  Over-Approximation of
  the Forward Reachable Set of the Closed-Loop System}%
\label{sec/method/reach}

Using the over-approximation \( \s{rob/ss/safe/oa} \) of
the robot's safe state space as a set of initial robot states,
in this section we employ reachability analysis techniques to
compute an over-approximation of the forward reachable set of
the closed loop system.
Since the accuracy of the over-approximation of
the forward reachable set depends on the partition \( \s{rob/ss/partition} \)
of the over-approximation of the safe set \( \s{rob/ss/safe/oa} \),
we also propose a method to refine the partition \( \s{rob/ss/partition} \)
in order improve the accuracy of
the over-approximation \( \s{rob/ss/frs/crit/oa} \).
This way we can control the tradeoff between computational complexity
affected by
the number of cells in \( \s{rob/ss/partition} \)
and
accuracy of the over-approximation \( \s{rob/ss/frs/crit/oa} \).

More specifically,
to obtain a tight over-approximation of
the set of states \( \s{rob/ss/frs/crit} \)
reachable from \( \s{rob/ss/safe/oa} \)
using the robot's closed-loop dynamics, we begin by noting that
the partition constructed by \autoref{alg/ss/partition} consists of
cells defined by \( \s{rob/state/dim} \)-dimensional intervals.
As such,
given a state space cell \( \s{rob/ss/cell} \in \s{rob/ss/safe/oa} \),
the state space cell \( \s{rob/ss/cell/frs/at}{\s{rob/ss/cell}} \)
reachable from \( \s{rob/ss/cell} \)
can be computed using \( \s{rob/f/oa} \) defined in~\eqref{eq/f/oa}
which returns a bounding box enclosing
the set of states reachable from \( \s{rob/ss/cell} \)
under the neural network controller, i.e.,
\begin{equation}
  \s{rob/ss/cell/frs/at}{\s{rob/ss/cell}; \s{nn/weight}, \s{nn/bias}} =
  \s{rob/f/oa/at}{\s{rob/ss/cell}}{
    \s{nn/oa/at}{\s{rob/ss/cell}; \s{nn/weight}, \s{nn/bias}}
  }
\end{equation}
where
\( \s{nn/oa} \) is a continuous map of
a \( \s{rob/state/dim} \)-dimensional interval to
a \( \s{rob/input/dim} \)-dimensional interval such that
\begin{eqnarray}
  \s{nn/at}{\s{rob/state}; \s{nn/weight}, \s{nn/bias}}
  \in
  \s{nn/oa/at}{\s{rob/ss/cell}; \s{nn/weight}, \s{nn/bias}}
  , ~~~
  \forall \s{rob/state} \in \s{rob/ss/cell}
  .
\end{eqnarray}
To obtain a function \( \s{nn/oa} \)
that returns a valid over-approximation of
the set of control inputs generated by the neural network \( \s{nn} \),
we employ the Interval Bound Propagation (IBP) method~\cite{50517380}.
In words,
for each cell \( \s{rob/ss/cell} \in \s{rob/ss/partition} \),
we employ the IBP method (function \( \s{nn/oa} \)) to compute bounds on
the robot's control inputs,
which we then propagate (function \( \s{rob/f/oa} \)) to obtain bounds on
the robot's forward reachable set.
Thus,
the over-approximation \( \s{rob/ss/frs/crit/oa} \) of
the set of states reachable by the closed-loop system after one step
can be obtained as
\begin{equation}%
  \label{eq/reach/crit}
  \s{rob/ss/frs/crit/oa} =
  \s{set/rel/union/multi}{\s{rob/ss/cell} \in \s{rob/ss/partition}}{}{
    \s{rob/ss/cell/frs/at}{\s{rob/ss/cell}; \s{nn/weight}, \s{nn/bias}}
  }
  .
\end{equation}

We remark that
the over-approximation error between
the bounds on the control inputs computed using IBP
and
the actual bounds on the control inputs
generated by \( \s{nn} \) for the states in \( \s{rob/ss/cell} \)
increases with the size of \( \s{rob/ss/cell} \),
as explained in~\cite{50517380}.
Therefore,
a fine partition \( \s{rob/ss/partition} \) of
the safe space over-approximation \( \s{rob/ss/safe/oa} \)
is generally required
in order to obtain a tight over-approximation of \( \s{rob/ss/frs/crit} \).
To refine \( \s{rob/ss/partition} \)
while keeping the total number of cells as low as possible,
we further subdivide only cells whose forward reachable set
intersects with the complement of \( \s{rob/ss/safe/oa} \).
To identify such cells,
we recall that
the forward reachable set 
\( \s{rob/ss/cell/frs/at}{\s{rob/ss/cell}} \) of \( \s{rob/ss/cell} \)
is a composite interval with the same dimension as \( \s{rob/ss/cell} \)
and
consider the two components
\( \s{rob/ss/cfg/cell}' \in \s{rob/ss/cfg} \)
and
\( \s{rob/ss/misc/cell}' \in \s{rob/ss/misc} \)
of \( \s{rob/ss/cell/frs/at}{\s{rob/ss/cell}} \) such that
\( \s{rob/ss/cell/frs/at}{\s{rob/ss/cell}} =
\s{rob/ss/cfg/cell}' \times \s{rob/ss/misc/cell}' \).
%
%
Following the procedure presented in \autoref{sec/method/feas},
we can compute
an over-approximation \( \s{rob/fpoa/at}{\s{rob/ss/cfg/cell}'} \)
of the robot's footprint
corresponding to
the forward reachable set of configurations \( \s{rob/ss/cfg/cell}' \).
Therefore,
a cell \( \s{rob/ss/cell} \) in \( \s{rob/ss/partition} \)
only needs to be subdivided if
\( \s{rob/fpoa/at}{\s{rob/ss/cfg/cell}'} \) intersects with
the complement \( \s{ws/cmpl} \) of the workspace \( \s{ws} \)
or
\( \s{rob/ss/misc/cell}' \) intersects with
the complement \( \s{rob/ss/misc/compl} \) of \( \s{rob/ss/misc} \).

\IncMargin{1em}
\begin{algorithm2e}
    \caption{Verify whether an outer approximation of the robot's one-step reachable set 
    violates safety requirements by more than margins $ \s{rob/ss/refine/vphresh},\ \s{rob/ss/refine/vqhresh}$ anywhere in cell $\s{rob/ss/cell}$.\label{alg/ss/refinement/test}}%
    \Fn{$A=$\Call{PenaltyTest}{$\s{rob/ss/cell},\ \s{rob/ss/refine/vphresh},\ \s{rob/ss/refine/vqhresh};\ \s{ws},\ \s{rob/body},\ \s{rob/ss/misc},\ f$}}{
    $ A,\ \s{rob/ss/cell}' \gets$ false, $\s{rob/ss/cell/frs/at}{\s{rob/ss/cell}}\ $\;
     \If{$\s{set/rel/vol}{ \s{rob/fpoa/at}{\s{rob/ss/cfg/cell}'} \s{set/rel/isect} \s{ws/compl} } >\s{rob/ss/refine/vphresh}\ \mathrm{\mathbf{or} }\ \s{set/rel/vol}{ \s{rob/ss/misc/cell}' \s{set/rel/isect} \s{rob/ss/misc/compl} }>\s{rob/ss/refine/vqhresh}$}
{
    $A\ \gets\ $ true \;
}
    }
\end{algorithm2e}
\DecMargin{1em}

\IncMargin{1em}
\begin{algorithm2e}
    \caption{Refine $\s{rob/ss/partition}$ and remove cells where the robot's one-step reachable set may violate safety requirements beyond specified tolerances.\label{alg/ss/refinement}}%
    \Fn{$\s{rob/ss/partition}'=$\Call{RefinePrtn}{$\s{rob/ss/partition},\ \s{rob/ss/partition/wthresh},\ \s{rob/ss/refine/vphresh},\ \s{rob/ss/refine/vqhresh};\ \s{ws},\ \s{rob/body},\ \s{rob/ss/misc},\ f$}}{
    %
    $\s{rob/ss/frontier},\ \s{rob/ss/partition}' \gets \s{rob/ss/partition},\ \emptyset $ \;
     \While{$\s{rob/ss/frontier}$ is not empty}
     {
        $\s{rob/ss/cell}\ \gets \text{a leaf cell in the last level of }\s{rob/ss/frontier}$ that has not yet been examined\;
        \uIf{ \Call{PenaltyTest}{$\s{rob/ss/cell},\ \s{rob/ss/refine/vphresh},\ \s{rob/ss/refine/vqhresh};\ \s{ws},\ \s{rob/body},\ \s{rob/ss/safe},\ f$} }
        {
            \uIf{\Call{MaxWidth}{$\s{rob/ss/cfg/cell}$}$ > \s{rob/ss/partition/wthresh}$}
            {
                $\s{rob/ss/frontier} \gets \s{rob/ss/frontier}\cup \left( \text{\Call{SubdivideCell}{$\s{rob/ss/cfg/cell}$}} \times \s{rob/ss/misc/cell} \right)$ \tcp*{Children $\s{rob/ss/cell}$}
            }
            \Else{
                $\s{rob/ss/frontier}\ \gets\ \s{rob/ss/frontier} \setminus \s{rob/ss/cell}$ \;
            }
        }
        \Else{
            $\s{rob/ss/cell}^{sib},\ \s{rob/ss/cell}^{par} \gets $ \Call{GetSibling}{$\s{rob/ss/cell}$}, \Call{GetParent}{$\s{rob/ss/cell}$}\;
            \uIf{\Call{PenaltyTest}{$\s{rob/ss/cell}^{sib},\ \s{rob/ss/refine/vphresh},\ \s{rob/ss/refine/vqhresh};\ \s{ws},\ \s{rob/body},\ \s{rob/ss/safe},\ f$}}{
            \uIf{\Call{MaxWidth}{$\s{rob/ss/cfg/cell}^{sib}$}$ < \s{rob/ss/partition/wthresh}$}
            {
                $\s{rob/ss/frontier} \gets \s{rob/ss/frontier}\cup \left({\text{\Call{SubdivideCell}{$\s{rob/ss/cfg/cell}^{sib}$}} \times \s{rob/ss/misc/cell}^{sib}} \right)$ \tcp*{Children of $\s{rob/ss/cell}^{sib}$}
            }
            \Else{
                $\s{rob/ss/frontier}\ \gets\ \s{rob/ss/frontier} \setminus \s{rob/ss/cell}$ \;
            }
            }
            \uElseIf{
                $\s{rob/ss/cell}^{sib}=\emptyset \ \mathrm{\mathbf{or}}$  \Call{CannotMerge}{$\s{rob/ss/cell},\s{rob/ss/cell}^{sib}$}
            }{
                $\s{rob/ss/frontier} \gets (\s{rob/ss/frontier}\setminus \{\s{rob/ss/cell},\s{rob/ss/cell}^{sib},\s{rob/ss/cell}_{par}\} )$ \;
                $\s{rob/ss/partition}' \gets \s{rob/ss/partition}' \cup \s{rob/ss/cell} \cup \s{rob/ss/cell}^{sib}$ \;
            }
            \Else{
                $\s{rob/ss/frontier} \gets (\s{rob/ss/frontier}\setminus \{\s{rob/ss/cell},\s{rob/ss/cell}^{sib},\s{rob/ss/cell}_{par}\} )\cup $ \Call{Merge}{$\s{rob/ss/cell}^{sib},\s{rob/ss/cell}$} \tcp*{Replace $\s{rob/ss/cell}_{par}$}
            }
        }
      }
    }
\end{algorithm2e}
\DecMargin{1em}

The process described above is
outlined in \autoref{alg/ss/refinement} which
adaptively subdivides cells in \( \s{rob/ss/partition} \)
with large over-approximation errors of their forward reachable sets.
Specifically,
for each cell
\( \s{rob/ss/cell} = (\s{rob/ss/cfg/cell}, \s{rob/ss/misc/cell}) \)
in
\( \s{rob/ss/partition} \),
we check whether
the area of \( \s{ws/compl} \) covered by
\( \s{rob/fpoa/at}{\s{rob/ss/cfg/cell}'} \)
or
the volume of \( \s{rob/ss/misc/compl} \) covered by
\( \s{rob/ss/misc/cell}' \)
are greater than user specified thresholds
\( \s{rob/ss/refine/vphresh} \) and \( \s{rob/ss/refine/vqhresh} \),
respectively. 
%
If these conditions hold and the size of the cell admits further subdivision,
then the cell gets split and replaced by smaller ones. 
%
Otherwise,
the sibling \( \s{rob/ss/cell}^{sib} \) of cell \( \s{rob/ss/cell} \)
is retrieved%
\footnote{
  Notice that, since each cell can be split only once into two subcells,
  \( \s{rob/ss/partition} \) can be represented as a binary tree.
}
and
if the volume of their parent cell \( \s{rob/ss/cell}_{mrg} \)
that lies outside the set of safe state \( \s{rob/ss/safe} \) is negligible,
then the cells \( \s{rob/ss/cell} \) and \( \s{rob/ss/cell}^{sib} \)
are merged
in order to reduce the size of \( \s{rob/ss/partition} \). 
%
Finally,
the over-approximation \( \s{rob/ss/frs/crit/oa} \) can be obtained as
the union of the forward reachable set of
each cell in \( \s{rob/ss/partition} \).
\autoref{alg/ss/refinement} effectively controls the tradeoff between
computational complexity affected by
the number of cells in the partition \( \s{rob/ss/partition} \)
and
accuracy of the over-approximation \( \s{rob/ss/frs/crit/oa} \).

In the following theorem we show that \autoref{alg/ss/refinement} terminates in finite time and its worst case run-time is related to the size of the workspace, the number of intervals in $\s{rob/ss/misc}$, and the reciprocal of the volume of the smallest allowable cell. Moreover, we show that the resulting cell partition $\s{rob/ss/partition}'$ inherits the safety assurances of $\s{rob/ss/partition}$ and that safety violations of its one-step reachable set are bounded by a user-specified tolerance.

\begin{theorem}\label{thm:FinalPartition}
Suppose \autoref{alg/ss/partition} is applied to a compact set $\s{ws} \s{set/rel/sub} \s{set/num/real}^{2}$ and let $\s{rob/ss/partition}=\s{rob/ss/partition}_{[p]} \times \s{rob/ss/misc}$ denote its output. Consider also the mapping $\s{rob/body}:\s{set/num/real}^{2}\times[0,2\pi) \mapsto \s{set/Realsets}(\s{set/num/real}^{2})$, the tolerance $\s{rob/ss/partition/wthresh}>0$, and assume that the assumptions of \autoref{thm:InitialPartition} hold. Suppose that \autoref{alg/ss/refinement} is applied to $\s{rob/ss/partition}$ with discrete-time forward dynamics governed by Equations \ref{eq/sys}-\ref{eq/nn/def/layer} and tolerances $\s{rob/ss/refine/vphresh},\s{rob/ss/refine/vqhresh}>0$. Then, \autoref{alg/ss/refinement} terminates after $\#\left( \s{rob/ss/misc}\right)\s{set/rel/vol}{\s{rob/ss/partition}_{[p]}}\s{rob/ss/partition/wthresh}^{-3}$ repetitions of its while loop and the output, $\s{rob/ss/partition}'$, satisfies 
$\cup_{\s{rob/ss/cell}\in \s{rob/ss/partition}'}\s{rob/ss/cell}  \subseteq  \s{rob/ss/safe}\oplus (\mathcal{C}_{\s{rob/ss/partition/wthresh}}\times \mathrm{0})$ and
$ \cup_{\s{rob/ss/cell}\in \s{rob/ss/partition}'}\s{rob/ss/cell}  \cup_{p\in \s{rob/ss/cfg/cell}} \s{rob/body}(p) \subseteq \s{ws}\oplus\mathrm{B}_{v}$. Moreover, if $\s{rob/ss/cell}'=\s{rob/ss/cell/frs/at}{\s{rob/ss/cell}}$ for $\s{rob/ss/cell}\in\s{rob/ss/partition}'$, then $\s{set/rel/vol}{ \s{rob/fpoa/at}{\s{rob/ss/cfg/cell}'} \s{set/rel/isect} \s{ws/compl} } \leq\s{rob/ss/refine/vphresh}$ and $\s{set/rel/vol}{ \s{rob/ss/misc/cell}' \s{set/rel/isect} \s{rob/ss/misc/compl} }\leq\s{rob/ss/refine/vqhresh}$.
\end{theorem}
\begin{proof}
Most assertions in this proof follow from observing that if $\s{rob/ss/cell}' \in \s{rob/ss/partition}'$, then $\s{rob/ss/cell}' \subseteq \s{rob/ss/cell}\in \s{rob/ss/partition}$, and by following the arguments in the proof of \autoref{thm:InitialPartition}. The bound on the one-step reachable set's safety violations holds because cells are only added to $\s{rob/ss/cell}'$ if \autoref{alg/ss/refinement/test} returns false. Only the complexity bound requires additional discussion.

To establish the complexity bound, observe that for any cell $\s{rob/ss/cell}'\in\s{rob/ss/frontier}$, we have that $\s{rob/ss/cell}'_{[p]} \subseteq \s{rob/ss/cell}_{p}\in \s{rob/ss/partition}_{[p]}$ and $\s{rob/ss/cell}'_{[q]}=\s{rob/ss/cell}_{q}\in \s{rob/ss/misc}$ because cells are only added to $\s{rob/ss/frontier}$ by subdividing their configuration dimensions. For any interval $\mathcal{Z}\in \s{rob/ss/misc}$, note that
\begin{align*}
    V_{\mathcal{Z}}=\s{set/rel/vol}{\bigcup_{\s{rob/ss/cell}\in \mathcal{X}_{\mathcal{Z}} } \s{rob/ss/cell}_{[p]} } ,
\end{align*}
where $\mathcal{X}_{\mathcal{Z}} = \left\{ \s{rob/ss/cell}\in \s{rob/ss/frontier} | \text{\Call{IsLeaf}{$\s{rob/ss/cell}$} = true, and }\s{rob/ss/cell}_{[q]}=\mathcal{Z} \right\}$ remains constant when cells are subdivided or merged, or decreases as cells are removed. Initially, 
$V_{\mathcal{Z}} =\s{set/rel/vol}{\s{rob/ss/partition}_{[p]}}$ for all $\mathcal{Z}\in \s{rob/ss/misc}$, which bounds $V_{\mathcal{Z}}$ above at all iterations. Observe that due to the subdivision operations, if $\s{rob/ss/cell}^1,\s{rob/ss/cell}^2\in\mathcal{X}_{\mathcal{Z}}$, then $\s{rob/ss/cfg/cell}^1$, and $\s{rob/ss/cfg/cell}^2$ can only overlap on their boundaries. Then \Call{MaxWidth}{$\s{rob/ss/cfg/cell}$}$ \geq \s{rob/ss/partition/wthresh}$ implies there can be at most $n_{\mathcal{Z}}= \s{set/rel/vol}{\s{rob/ss/partition}_{[p]}}\s{rob/ss/partition/wthresh}^{-3} $ cells in $\mathcal{X}_{\mathcal{Z}}$, and subsequently at most $L_{max}=\#\left( \s{rob/ss/misc}\right)\s{set/rel/vol}{\s{rob/ss/partition}_{[p]}}\s{rob/ss/partition/wthresh}^{-3} $ leaves in $\s{rob/ss/frontier}$. At each repetition of the while loop, a leaf cell, $\mathcal{C}\in \s{rob/ss/frontier}$, is examined and one of the following happens:
\begin{itemize}
    \item the number of leaves increases by one because $\mathcal{C}$ or its sibling is subdivided; or
    \item the number of leaves decreases by one or more because either $\mathcal{C}$, its sibling, and/or its parent are removed, or because $\mathcal{C}$ and its sibling are merged and replace their parent.
\end{itemize}
The required bound is found by observing that leaves can only be added at most $L_{max}$ times and removed at most $L_{max}$ times, which completes the complexity argument.
\end{proof}
\subsection{Safety-Aware Control Training}%
\label{sec/method/optim}

This section uses the over-approximation of the forward reachable set of
the closed loop system \( \s{rob/ss/frs/crit/oa} \),
to design appropriate penalty terms which, 
when added to the loss function~\eqref{eq/optim/base/costfun}, 
minimize the overlap between this forward reachable set and
the set of unsafe states.
In this way,
we can reduce safety violations after re-training of
the neural network controller~\eqref{eq/nn/def}.
To do so,
we solve the following optimization problem at
every iteration of \autoref{alg/main} to update the neural network parameters
%
\begin{equation}%
  \label{eq/optim/aug}
  \begin{aligned}
    & \underset{
      \s{nn/weight}, \s{nn/bias}
    }{\text{minimize}}
    & & \s{ml/aug/costfun/at}{\s{nn/weight}, \s{nn/bias}} \\
  \end{aligned}
\end{equation}
\begin{equation}
  \label{eq/optim/aug/costfun}
  \s{ml/aug/costfun/at}{\s{nn/weight}, \s{nn/bias}} =
  \s{ml/base/costfun/at}{\s{nn/weight}, \s{nn/bias}}
  +
  \s{ml/aug/coef/safety}
  \cdot
  \s{ml/aug/safety/penalty/at}{\s{nn/weight}, \s{nn/bias}},
\end{equation}
where
\( \s{ml/aug/coef/safety} \) is a positive constant
and
\( \s{ml/aug/safety/penalty} \)
a penalty term that measures safety violations.
The goal in designing the penalty term \( \s{ml/aug/safety/penalty} \)
is to push
the over-approximation of the new forward reachable set \( \s{rob/ss/frs/crit/oa} \)
inside the under-approximation
\( \s{rob/ss/safe/ua} =
\s{set/rel/union}_{\s{rob/ss/cell} \in \s{rob/ss/partition/safe}}{
  \s{rob/ss/cell} }
\)
of the set of safe states \( \s{rob/ss/safe} \),
where
\(
\s{rob/ss/partition/safe}
\)
denotes the subset of the partition \( \s{rob/ss/partition} \)
consisting of only safe cells.
To do so,
we define the penalty term \( \s{ml/aug/safety/penalty} \) as
%
$
  \label{eq/ml/aug/safety/penalty}
  \begin{aligned}
    \s{ml/aug/safety/penalty/at}{\s{nn/weight}, \s{nn/bias}}
    =
    \sum\nolimits_{\s{rob/ss/cell}' \in \s{rob/ss/partition/frs}}{
      \s{ml/aug/metric/at}{
        \s{rob/ss/cell}', \s{rob/ss/partition/safe}
      }
    }^{2}
    ,
  \end{aligned}
$
where
\(
\s{rob/ss/partition/frs}
=
\s{set/def}{ \s{rob/ss/cell/frs/at}{\s{rob/ss/cell}} }{
  \s{rob/ss/cell} \in \s{rob/ss/partition}
}
\)
and
\( \s{ml/aug/metric/at}{\s{rob/ss/cell}, \s{rob/ss/partition/safe}} \)
is a valid metric of the volume occupied by
the intersection of 
the composite interval \( \s{rob/ss/cell} \)
and
\(
\s{set/rel/union}_{\s{rob/ss/cell} \in \s{rob/ss/partition/safe}}{
  \s{rob/ss/cell}
}
\).
Thus,
%
%
\( \s{ml/aug/safety/penalty} \) vanishes only if
\(
\s{rob/ss/cell}' \s{set/rel/subeq}
\s{set/rel/union}_{\s{rob/ss/cell} \in \s{rob/ss/partition/safe}}{
  \s{rob/ss/cell}
}
\)
for all \( \s{rob/ss/cell}' \in \s{rob/ss/partition/frs} \),
i.e.,
the set of safe states is rendered invariant under the closed-loop dynamics.
Note that the overlap between the forward reachable set and
the set of unsafe states
\( \s{rob/ss/frs/crit} \s{set/rel/diff} \s{rob/ss/safe} \)
also provides numerical bounds on possible safety violations by
the closed loop system,
that can be used as a measure of reliability of
the neural network controller~\eqref{eq/nn/def}.

\section{Numerical Experiments}%
\label{sec/sims}
In this section, we illustrate the proposed method through numerical experiments involving a rectangular shaped robot that operates inside
the workspace depicted in \autoref{fig/sim1/ws}.
\begin{figure}
  \centering
  \includegraphics[width=0.9\linewidth]{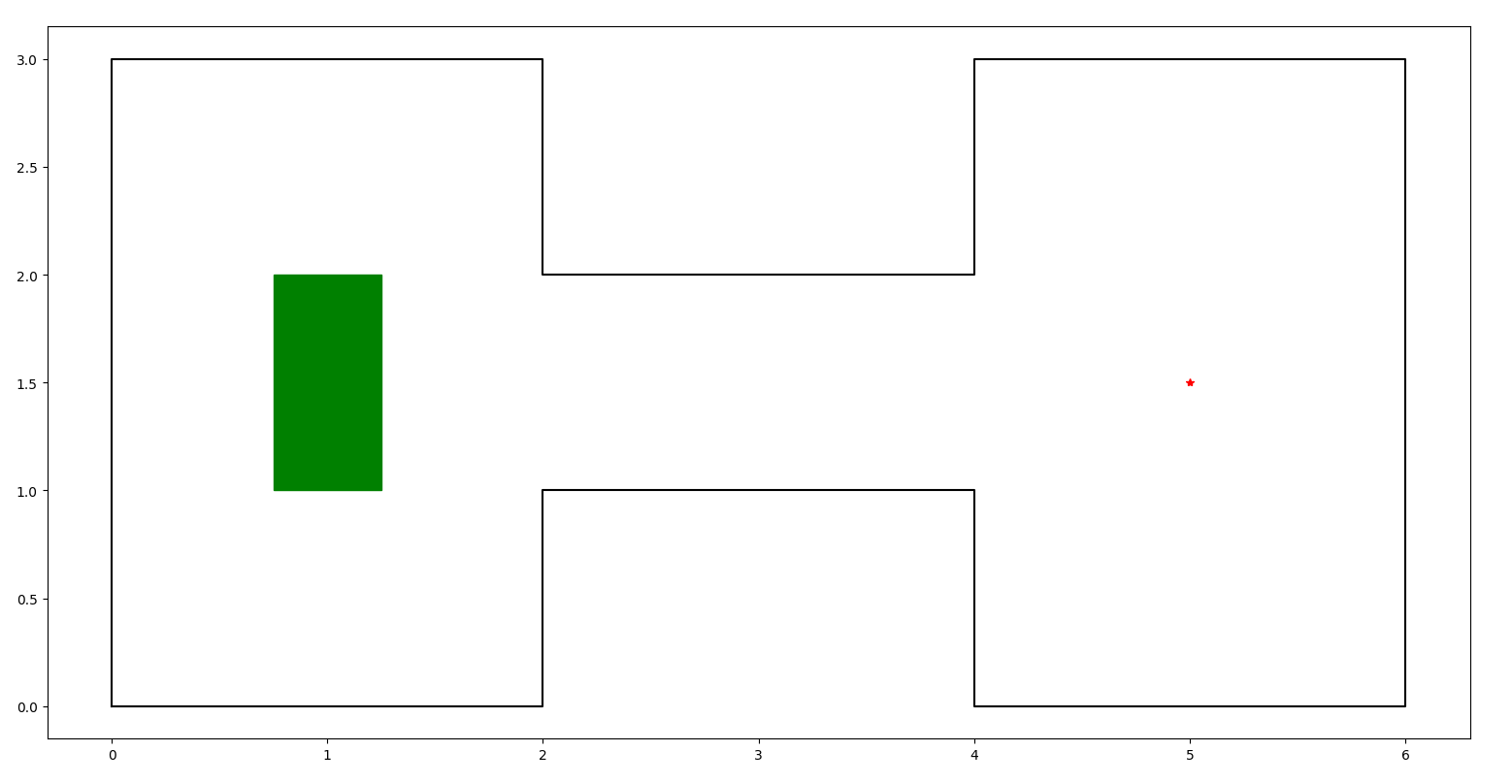}
  \caption{
    Workspace, consisting of two rooms and a corridor, and
    the robot's body at default orientation (green rectangle).
    The desired position of the robot's reference point is placed at
    the room on the right (red point).
  }%
  \label{fig/sim1/ws}
\end{figure}
The robot is assumed to be holonomic obeying the simple kinematic model
\begin{equation}
  \s{rob/state/at}{k+1} =
  \s{rob/state/at}{k} + K \cdot \s{nn/at}{\s{rob/state/at}{k}}
  ,
\end{equation}
with \( \s{rob/state} = \s{rob/cfg} \) and \( K = 0.01 \).
The dataset \( \s{data/set} \)
was assembled by computing 500 feasible trajectories
using the RRT$^{\star}$ algorithm~\cite{75318386}
and initializing the robot at random configurations.
Particularly,
each point \( (\s{rob/state}, \s{rob/input}) \in \s{data/set} \)
corresponds to a sampled trajectory point
\( (\s{rob/state}, \s{rob/input}') \),
where the control input \( \s{rob/input}' \) is modified to ensure that the implicitly defined vector field vanishes at
the goal configuration, i.e.,
\begin{equation}
  \s{rob/input} =
  10 \cdot
  \frac{\s{rob/input}'}{ \| \s{rob/input}' \| }
  \frac{ \| \s{rob/state/goal} - \s{rob/state} \| }{
    1 - \| \s{rob/state/goal} - \s{rob/state} \|
  }.
\end{equation}
The over-approximation of the robot's forward reachable set for
a given cell \( \s{rob/ss/cell} \) is computed as
the Minkowski sum between that cell and
the controller's output reachable set for that cell,
i.e.,
\begin{equation}
  \s{rob/f/oa/at}{\s{rob/ss/cell}}{\s{nn/oa/at}{\s{rob/ss/cell}}} =
  \s{rob/ss/cell}
  \oplus
  \s{lrp}{K \cdot \s{nn/oa/at}{\s{rob/ss/cell}}}
  .
\end{equation}
%
%
Additionally,
the metric \( \s{ml/aug/metric} \) used
in the evaluation of the penalty term \( \s{ml/aug/safety/penalty} \)
is defined by
\begin{equation}
	\label{eqn:V_function}
  \s{ml/aug/metric/at}{\s{rob/ss/cell}, \s{rob/ss/partition/safe}} =
  \s{lrp}{\s{set/rel/vol/scaled}{\s{rob/ss/cell}}{\delta}}^{1/3}
  -
  \s{lrp}{
    \sum_{\s{rob/ss/cell}' \in \s{rob/ss/partition/safe}}{
      \s{set/rel/vol/scaled}
      {\s{rob/ss/cell} \s{set/rel/isect} \s{rob/ss/cell}'}
      {\delta}}
  }^{1/3}
  ,
\end{equation}
where
\(
\delta = [
\delta_{\s{rob/pos/x}} = 1,
\delta_{\s{rob/pos/y}} = 1,
\delta_{\s{rob/angle}} = 1 / 2\pi
]
\)
are constant scaling factors.
Particularly,
the cubic root of the volume of composite intervals in \eqref{eqn:V_function} is used because, in practice, it prevents terms corresponding to slim cells from vanishing too fast when, e.g.,  only one of their dimensions has grown small whereas
the others have not.
The scaling factors 
\( \delta_{\s{rob/pos/x}}, \delta_{\s{rob/pos/y}}, \delta_{\s{rob/angle}} \)
allow to assign different weights to each dimension,
given that they have incompatible units
(\(\s{rob/pos/x}\) and \(\s{rob/pos/y}\) are in meters
and
\( \s{rob/angle} \) is given in rad).



\begin{table}
  \begin{center}
    \begin{tabular}{l|rr|rr}
      \hline
      NN & Shape & \(\epsilon_w\) & \multicolumn{2}{c}{Partition Size} \\
         &&& Initial & Refined \\
      \hline
      \(\phi_{1}\) & 3x50x50x50x3 & (0.1, 0.1, 0.2\(\pi\))   & 5434 & 6990 \\
      \(\phi_{2}\) & 3x50x50x50x3 & (0.25, 0.25, 0.2\(\pi\)) & 1046 & 1180 \\
      \(\phi_{3}\) & 3x50x50x3    & (0.1, 0.1, 0.2\(\pi\))   & 5434 & 5536 \\
      \hline
    \end{tabular}
  \end{center}
  \caption{
    Shape and subdivision thresholds used in the re-training of each
    neural network controller, along with the number of cells in
    the corresponding partitions \( \s{rob/ss/partition} \)
    after their construction and initial refinement.
  }
  \label{tbl/sim1/defs}
\end{table}



\begin{table}
\begin{center}
\begin{tabular}{l|rr|rr}
\hline
NN & \multicolumn{2}{c}{\(\lambda_{S}\)} & \multicolumn{2}{c}{\(J(w,b)\)} \\
& Step Size & Final Value & Initial Value & Final Value \\
\hline
\(\phi_{1}\) & \(1.0e{-4}\) & \(5.0e{-3}\) & 3.4200 & 3.6659 \\
\(\phi_{2}\) & \(2.0e{-4}\) & \(1.0e{-2}\) & 3.4200 & 3.6114 \\
\(\phi_{3}\) & \(4.0e{-5}\) & \(2.0e{-3}\) & 3.5345 & 3.7948 \\
\hline
\end{tabular}
\end{center}
  \caption{
    Step size and final value of
    the penalty coefficient \( \s{ml/aug/coef/safety} \) as well as
    the value of the loss function \( \s{ml/base/costfun} \)
    before and after re-training for each respective neural network controller.
  }
  \label{tbl/sim1/evol}
\end{table}

Three separate neural network controllers
\( \s{nn}_1, \s{nn}_2, \s{nn}_3 \)
with various shapes
were trained to safely steer the robot to the desired configuration
using the proposed methodology using different values for
the subdivision threshold \( \s{rob/ss/partition/wthresh} \).
The activation functions of the hidden layers of all three neural network controllers were selected to be hyperbolic tangent functions. The shapes of these neural networks and the values of the subdivision thresholds $\epsilon_w$  used for safe training of each controller can be seen in Table~\ref{tbl/sim1/defs}. Also, the parameters
\( \s{ml/base/coef/err} \) and \( \s{ml/base/coef/reg} \)
in~\eqref{eq/optim/base/costfun}
were set to
\( 1 / |\s{data/set}| \) and
\( 1 / (|\s{nn/weight}| + |\s{nn/bias}|) \),
respectively,
with
\( |\s{data/set}| \) denoting the number of data points in \( \s{data/set} \)
and
\( |\s{nn/weight}|, |\s{nn/bias}| \) denoting
the dimension of the corresponding neural network controller's
weight and bias vectors.
After initially training each network using the data in $\mathcal{D}$, 
we used \autoref{alg/ss/partition}
to obtain an initial partition \( \s{rob/ss/partition} \) of
the robot's configuration space
with the subdivision threshold
\( \s{rob/ss/partition/wthresh} =
\s{lrp}{
  {\s{rob/ss/partition/wthresh}}_x, {\s{rob/ss/partition/wthresh}}_y,
  {\s{rob/ss/partition/wthresh}}_{\theta}} \)
for each respective dimension.
Then,
we applied
\autoref{alg/ss/refinement} to refine the initial partition
\( \s{rob/ss/partition} \) based on the system's forward reachable sets
with
\( \s{rob/ss/refine/vphresh} = 1e-2 \).
The number of cells in the partition \( \s{rob/ss/partition} \)
associated with each neural network controller before and after
the first refinement can be seen in \autoref{tbl/sim1/defs}.
We remark that
blindly partitioning
the configuration space \( \s{ws} \times \s{set/num/circle} \) into cells
with dimensions \( 0.1 \times 0.1 \times 0.2\pi \)
would result in \( 14000 \) cells or more,
which implies that the proposed method requires
\( 1 - 6990 / 14000 \approx 50 \% \)
and
\( 1 - 5536 / 14000 \approx 60 \% \)
less subdivisions
for the cases of \( \s{nn}_1 \) and \( \s{nn}_3 \),
respectively.
Among these cells,
the active ones (i.e.,
cells whose forward reachable sets violate the safety specifications)
were used to evaluate the penalties at each iteration.
Finally,
50 iterations of the refinement and parameter update steps
of \autoref{alg/main} were performed.
Similarly to the approach proposed in~\cite{68156223}
where the size of the \(\epsilon\)-neighborhoods was gradually increased
to improve numerical stability of the retraining procedure,
the parameter \( \s{ml/aug/coef/safety} \) in~\eqref{eq/optim/aug/costfun}
was initialized at \(0\)
and
its value was increased
at the beginning of each iteration of \autoref{alg/main}
until a target value was reached at the 50th iteration.
The selected step size and
the corresponding target value of the parameter \( \s{ml/aug/coef/safety} \)
used in the re-training of each controller are listed in
\autoref{tbl/sim1/evol}.

\autoref{fig/sim1/vol} shows the evolution of
the total volume of \( \s{rob/ss/frs/crit/oa} \)
lying outside the under-approximation of the set of safe states,
which is given by
\( \s{rob/ss/partition/safe} \), i.e.,
\(
\s{set/rel/vol}{
  \s{lrp}{
    \s{set/rel/union}_{\s{rob/ss/cell} \in \s{rob/ss/partition}}{
      \s{rob/ss/cell/frs/at}{\s{rob/ss/cell}}
    }
  }
  \s{set/rel/isect}
  \s{lrp}{
    \s{set/rel/union}_{\s{rob/ss/cell} \in \s{rob/ss/partition/safe}}{
      \s{rob/ss/cell}
    }
  }
}
\)
whereas
\autoref{fig/sim1/ace} depicts the number of active cells at
the end of each epoch, with epoch \(0\) corresponding to
the initial training of the controllers (\( \s{ml/aug/coef/safety} = 0 \)).
After re-training,
the volume of  $\bar{\mathcal{Z}}_c$,
which bounds the volume of initial states that may escape the safe set
under the closed-loop dynamics,
and the number of active cells decreased by
\( 84.6 \%, 54.4 \%, 30.9 \%, \)
and
\( 31.6 \%, 15.3 \%, 9.5 \% \),
respectively,
for each of the three neural network controllers.
Additionally,
the values of the loss function \( \s{ml/base/costfun} \)
before the first epoch and after the last one
are listed in \autoref{tbl/sim1/evol}.
We note that
in all three cases
the value of the loss function \( \s{ml/base/costfun} \) increases,
which is to be expected as the proposed framework prioritizes
safety over liveliness of the neural network controller.
%

The projections of the cells in each partition
onto the \(\s{rob/pos/x}\s{rob/pos/y}\)-plane along with
the over-approximations of their forward reachable sets
before and after retraining can be seen in
\autoref{fig/sim1/rs/nn1}-\autoref{fig/sim1/rs/nn3}.
We remark that
the area covered by
the over-approximation of the robot's forward reachable set (colored in red)
that lies outside the set of initial states (colored in blue)
has noticeably decreased at the end of \( 50 \)-th epoch,
especially for the case of the neural network controller \( \s{nn}_1 \).
%
%
%
This area provides spatial bounds on possible safety violations induced by the trained controllers. As a result, while our method allows for safety violations, it does so with grace, by providing such bounds that can be controlled during training. Finally, 
several trajectories executed by the robot equipped with
the controller \( \s{nn}_1 \) before and after 50 iterations of re-training
are illustrated in \autoref{fig/sim1/traj}.
%
As can be seen in
subfigure~\ref{fig/sim1/traj/1},
the proposed re-training method results in safer robot trajectories than
the ones returned by the initially trained neural network controller,
whereas the behavior of the closed-loop system at states away from
the boundary of the safe set is not noticeably affected,
as indicated by subfigure~\ref{fig/sim1/traj/5}.

\begin{figure}
\vspace{-2mm}
  \centering
  \includegraphics[width=0.9\linewidth]{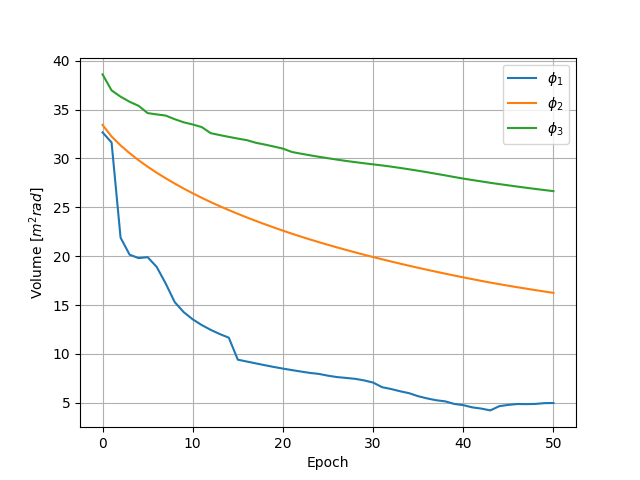}
  \vspace{-2mm}
  \caption{
    Total volume of active cells that lies outside the set of safe states
    at the beginning of each epoch.
  }
  \label{fig/sim1/vol}
\end{figure}

\begin{figure}
\vspace{-2mm}
  \centering
  \includegraphics[width=0.9\linewidth]{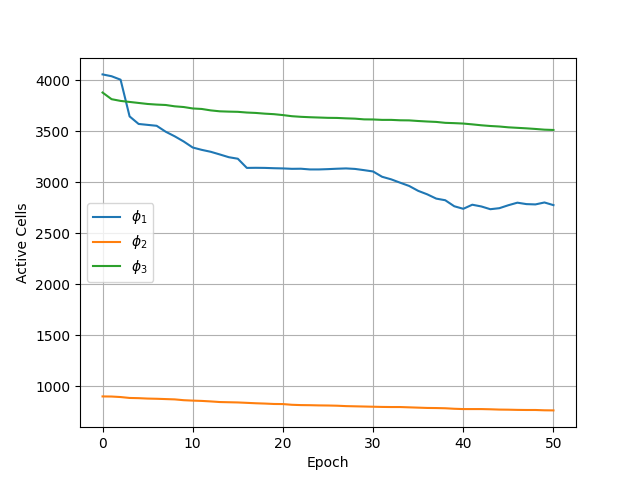}
  \vspace{-2mm}
  \caption{Number of active cells at the beginning of each epoch.}
  \label{fig/sim1/ace}
\end{figure}

\begin{figure}
  \centering
	  \subfigure[\( \s{nn}_1 \): Before re-training]{%
		\includegraphics[width=0.9\linewidth,height=100px]{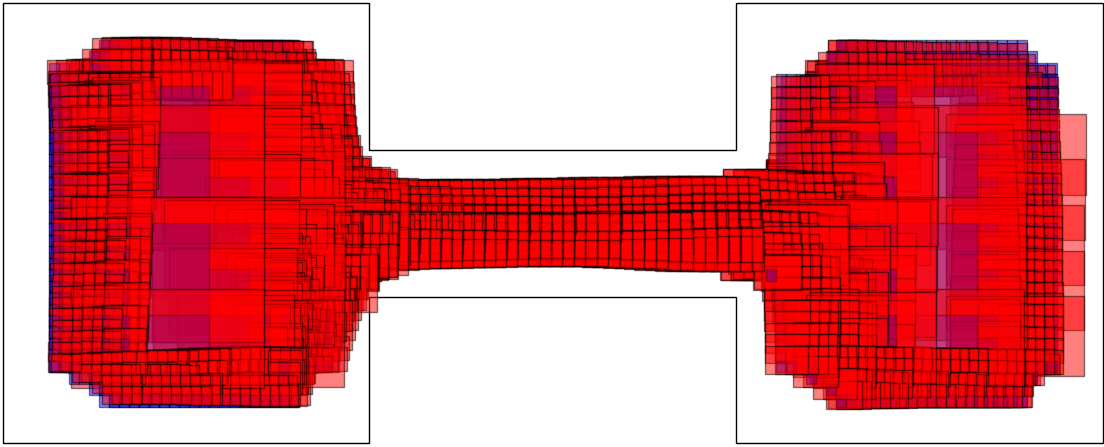}}%
  \hfill
  \subfigure[\( \s{nn}_1 \): After re-training]%
  {\includegraphics[width=0.9\linewidth,height=100px]{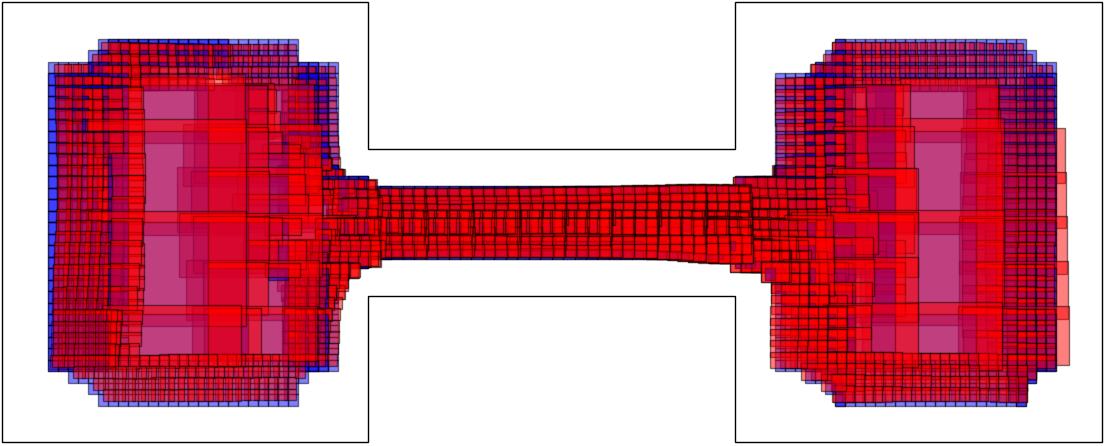}}
  \vspace{-2mm}
  \caption{
    Projections of cells (blue boxes) of partition corresponding to
    neural network controller \( \s{nn}_1 \)
    onto the \(\s{rob/pos/x}\s{rob/pos/y}\)-plane overlayed with
    the over-approximations of their forward reachable sets (red boxes)
    before and after retraining.
  }%
  \label{fig/sim1/rs/nn1}
\end{figure}

\begin{figure}
  \centering
  \subfigure[\( \s{nn}_2 \): Before re-training]%
	{\includegraphics[width=0.9\linewidth,height=100px]{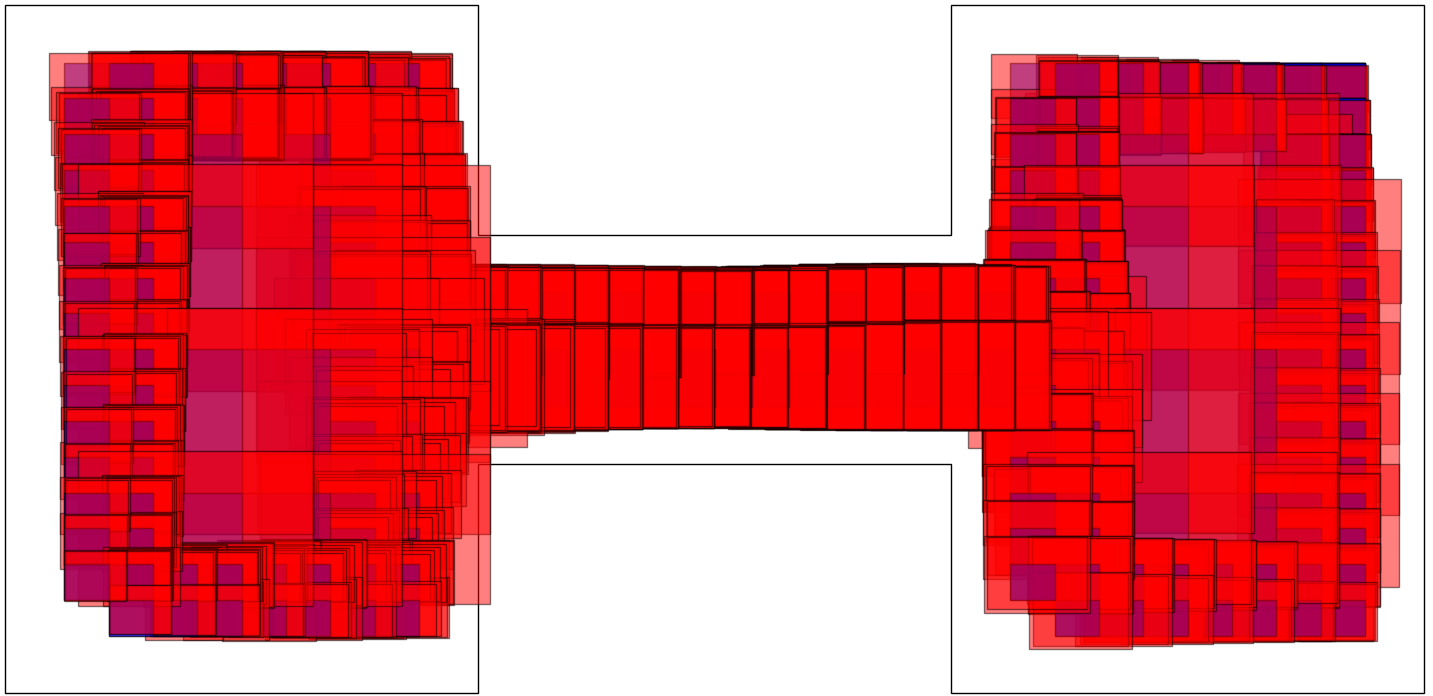}}
  \hfill
  \subfigure[\( \s{nn}_2 \): After re-training]%
	{\includegraphics[width=0.9\linewidth,height=100px]{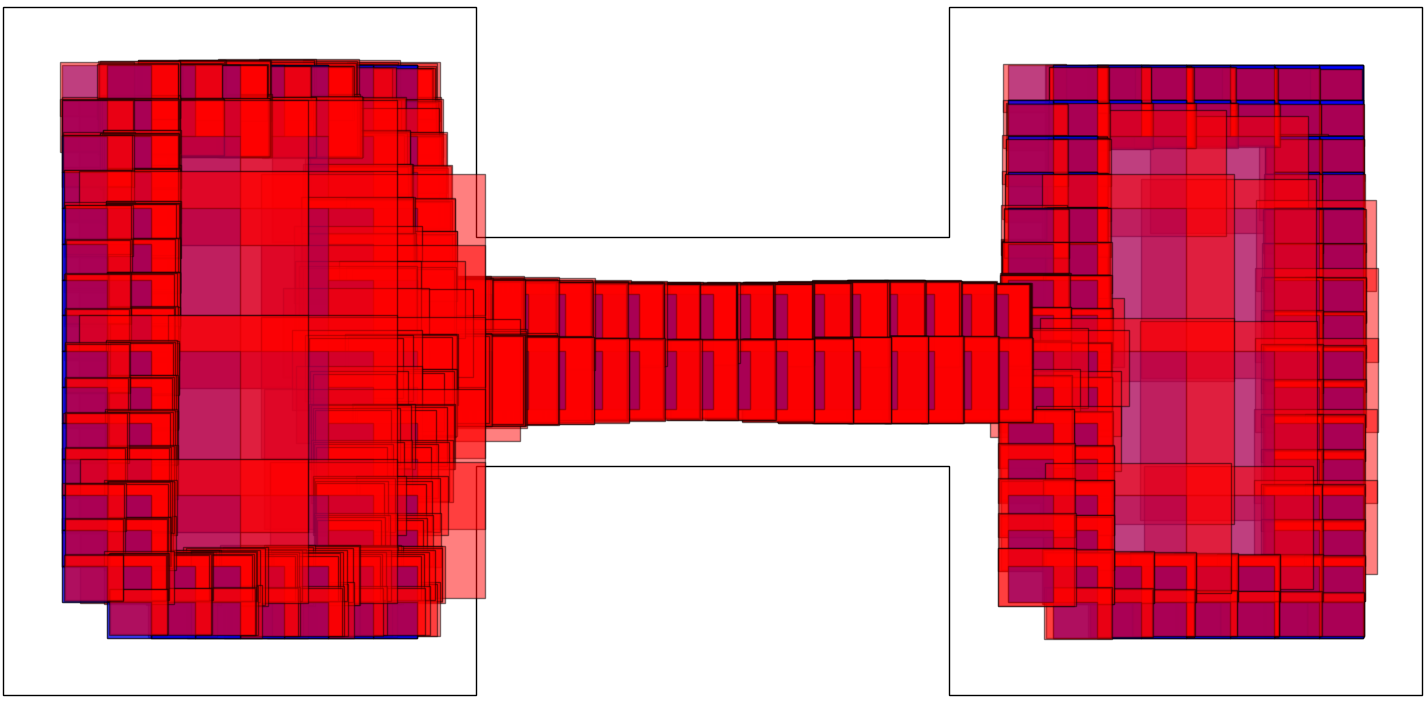}}
  \vspace{-2mm}
  \caption{
    Projections of cells (blue boxes) of partition corresponding to
    neural network controller \( \s{nn}_2 \)
    onto the \(\s{rob/pos/x}\s{rob/pos/y}\)-plane overlayed with
    the over-approximations of their forward reachable sets (red boxes)
    before and after retraining.
  }%
  \label{fig/sim1/rs/nn2}
\end{figure}

\begin{figure}
  \centering
  \subfigure[\( \s{nn}_3 \): Before re-training]%
    {\includegraphics[width=0.9\linewidth,height=100px]{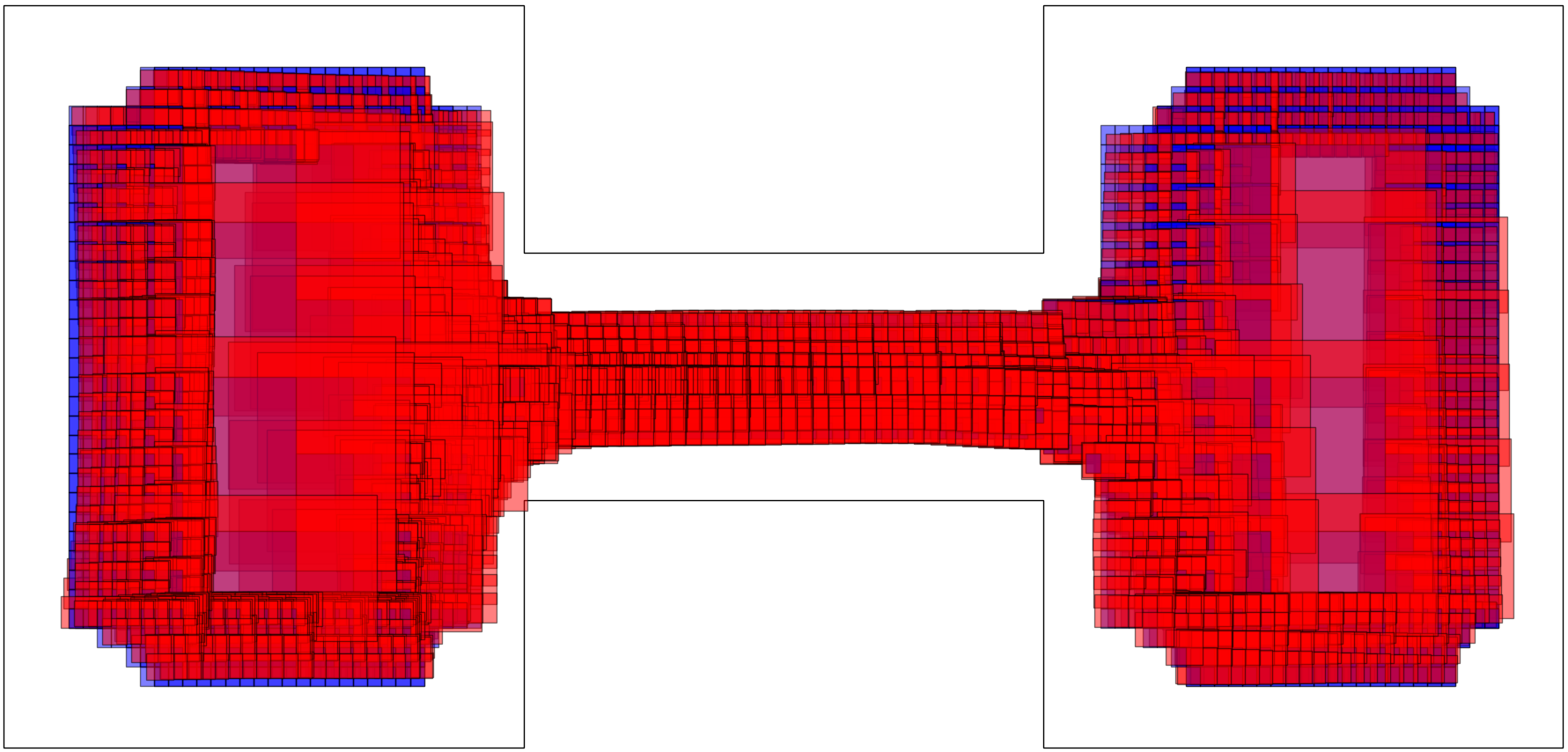}}
  \hfill
  \subfigure[\( \s{nn}_3 \): After re-training]%
    {\includegraphics[width=0.9\linewidth,height=100px]{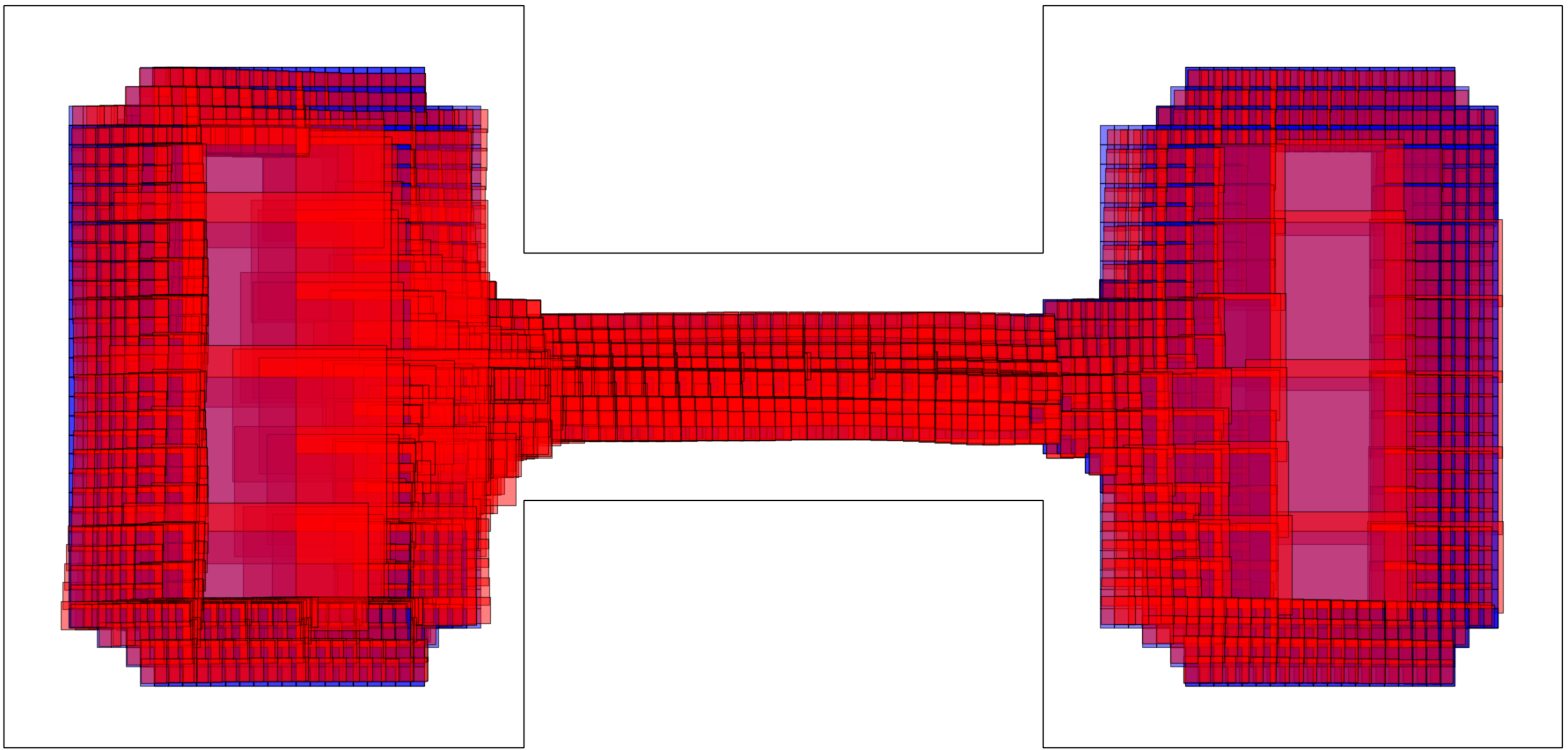}}
  \vspace{-2mm}
  \caption{
    Projections of cells (blue boxes) of partition corresponding to
    neural network controller \( \s{nn}_3 \)
    onto the \(\s{rob/pos/x}\s{rob/pos/y}\)-plane overlayed with
    the over-approximations of their forward reachable sets (red boxes)
    before and after retraining.
  }%
  \label{fig/sim1/rs/nn3}
\end{figure}

\begin{figure*}
  \centering
  \centering
  \subfigure[\label{fig/sim1/traj/1}]%
    {\includegraphics[width=0.9\linewidth]{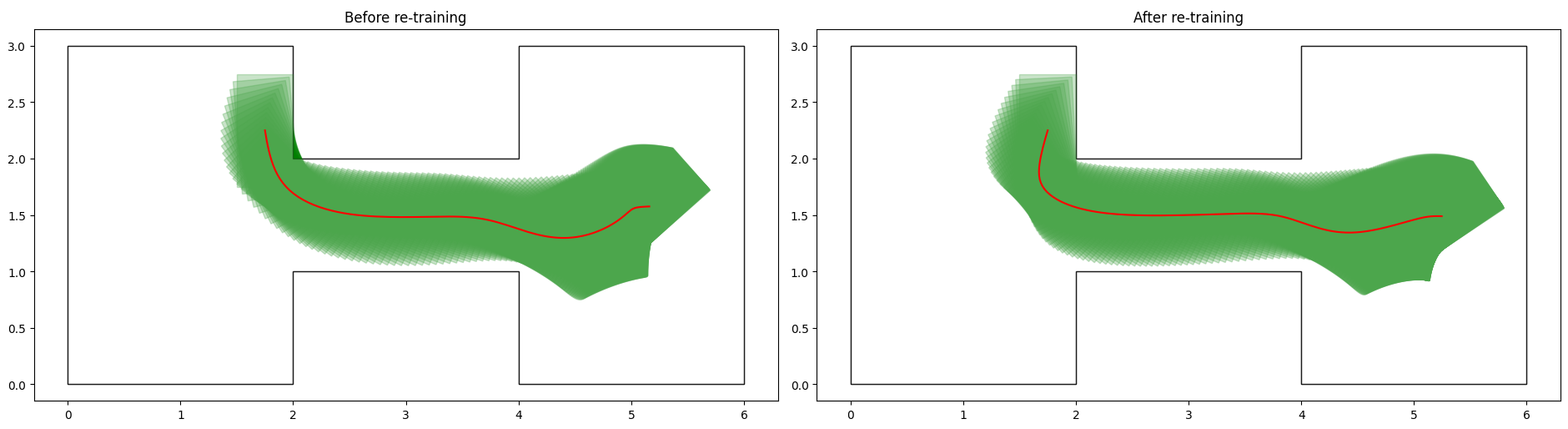}}
  \subfigure[\label{fig/sim1/traj/5}]%
	{\includegraphics[width=0.9\linewidth]{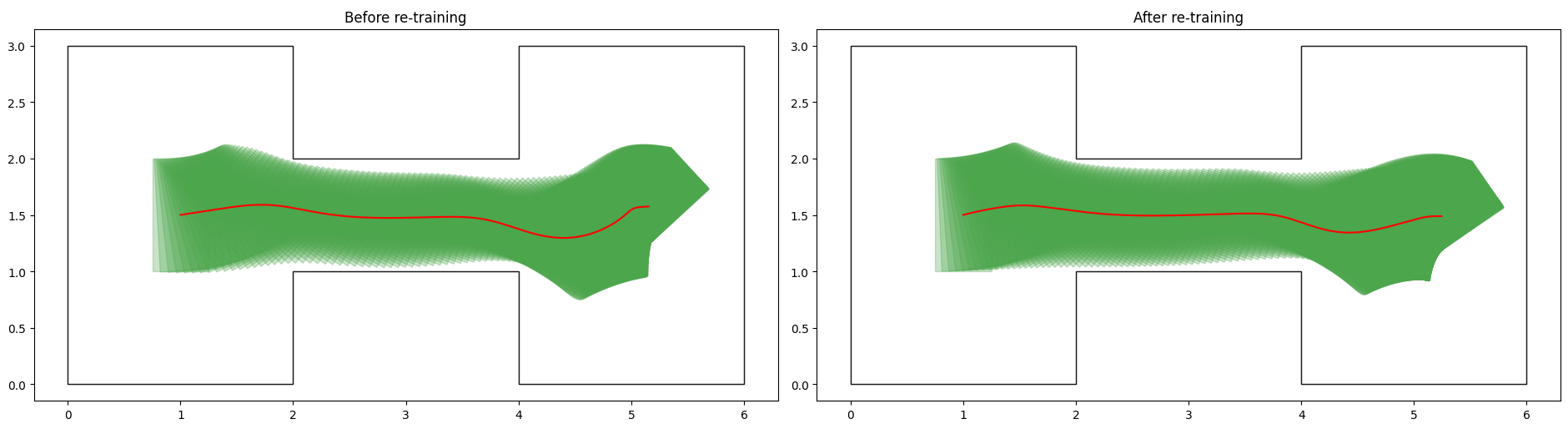}}
  \vspace{-2mm}
  \caption{
    Trajectories of the robot starting from various configurations
    before and after re-training of the controller \( \s{nn}_1 \).
    The path traversed by the robot's reference point is depicted in red
    and the body of the robot at various configurations is shown in green.
  }%
  \label{fig/sim1/traj}
\end{figure*}

\section{Conclusions}%
\label{sec/concl}

In this work,
we addressed the problem of safely steering a polygonal robot
operating inside a compact workspace
to a desired configuration using a feed-forward neural network controller
trained to avoid collisions between the robot and the workspace boundary. 
%
%
Compared to existing methods that depend strongly on the density of data points close to the boundary of the safe state space to train neural network controllers with closed-loop safety guarantees, our approach lifts such strong assumptions on the data that are hard to satisfy in practice and instead allows for {\em graceful} safety violations, i.e., of a bounded magnitude that can be spatially controlled.
%


\bibliographystyle{IEEEtran}
\bibliography{references.bib}
\end{document}